\documentclass[11pt,reqno]{amsart}

\usepackage{amsmath, amsfonts, amssymb,  mathrsfs,  array, stmaryrd,  indentfirst, amsthm,  hyperref}

\usepackage[margin=1.0in]{geometry}
\usepackage{setspace}
\usepackage{color}
\numberwithin{equation}{section}
\theoremstyle{plain}

\newtheorem{lemma}{Lemma}[section]
\newtheorem{theorem}{Theorem}[section]

\newtheorem{corollary}{Corollary}[section]
\newtheorem{assumption}{Assumption}[section]
\newtheorem{remark}{Remark}[section]

\makeatletter
\def\@setcopyright{}
\def\serieslogo@{}
\makeatother

\begin{document}
 
\author{Xiang Yu}
\address{Xiang Yu, Department of Mathematics, University of Michigan, 530 Church Street, Ann Arbor, MI 48109, USA}
\email{xymath@umich.edu}

\title[]{An Explicit Example of Optimal Portfolio-Consumption Choices with Habit Formation and Partial Observations.}

\begin{abstract}
We consider a model of optimal investment and consumption with both habit formation and partial observations in incomplete It\^{o} processes market. The investor chooses his consumption under the addictive habits constraint while only observing the market stock prices but not the instantaneous rate of return. Applying the Kalman-Bucy filtering theorem and the Dynamic Programming arguments, we solve the associated Hamilton-Jacobi-Bellman (HJB) equation explicitly for the path dependent stochastic control problem in the case of power utilities. We provide the optimal investment and consumption policies in explicit feedback forms using rigorous verification arguments.
\end{abstract}

\date{\today.}


\maketitle
\pagestyle{headings}



\section{Introduction}
Habit formation has become a popular choice for modeling preferences on consumption streams during recent years. It has been observed that the time separable property of von Neumann-Morgenstern utility is not consistent with some empirical experiments, for instance, the Equity Premium Puzzle. (See \cite{Mehra} and \cite{constantinides1988habit}). Therefore, the literature has been arguing that the past consumption pattern enforces a continuing impact on individual's current consumption decisions and the preference should depend on the consumption path. In particular, the linear habit formation preference $\mathbb{E}[\int_0^TU(t,c_t-Z_t)dt]$ has been widely accepted, where the index $Z_t$ stands for the accumulative consumption history. The utility function is decreasing in $Z_t$ which indicates that an increase in consumption today increases current utility but depresses all future utilities through the induced increase in future standards of living.

The continuous time utility maximization problem with habit formation in the complete It\^{o} processes market has been extensively studied in the past decades, see for instance, \cite{detemple92}, \cite{Schroder01072002}, \cite{Munk} and \cite{Eng09}. In incomplete markets, recently, this problem has been solved in the semimartingale market by \cite{Yu} and in the market with transaction costs and unbounded random endowment by \cite{Yu3}.

The contributions of the present work are two-fold. First, from the modeling perspective, we are considering some realistic information constraints to the investor. To be more precise, we are facing the case that the investor develops his consumption habits and meanwhile has only access to the public stock price information $\mathcal{F}_t^S=\sigma\{S_u:0\leq u\leq t\}$. Second, on the mathematical level, we reveal that this path-dependent stochastic control problem under the partial observation filtration $\mathcal{F}_t^S$ is actually easier than the case under full information. Indeed, it is well known that in incomplete markets with full observations, the work on the structure of the optimal strategies with habit formation is not promising. By restricting to the smaller filtration $\mathcal{F}_t^S$, however, the problem is very simple. We are able to solve the corresponding HJB equation fully explicitly. As a consequence, we can derive the $\mathcal{F}_t^S$-adapted optimal investment and consumption policies in feedback forms via rigorous verification arguments. Our analytical approach also allows us to avoid proving the Dynamic Programming Principle. 

Optimal investment problems under incomplete information have been studied by many authors recently, see among \cite{Lakner}, \cite{BS}, \cite{Brennan}, \cite{XiaY}, \cite{MR2648470} and \cite{MR2606929}. In this paper, we aim to combine the information constraint together with the addictive habit formation constraint and derive the unanticipated explicit forms of optimal strategies. The mathematical verification arguments can also be applied to other stochastic control problems in general.

The structure of the present paper is outlined as: Section $\ref{sec1}$ introduces the market model and the habit formation process. The utility maximization problem with addictive habit formation and partial observations is defined in Section $\ref{sec2}$. By applying the Kalman Bucy filtering theorem and Dynamic Programming arguments, we formally derive the HJB equation for the power utilities. We provide the decoupled form solution for this nonlinear PDE which reduces the algorithm to solving some auxiliary ODEs with constant coefficients. Section $\ref{sec3}$ contains rigorous proofs of the verification theorem. At last, four cases of fully explicit solutions of some auxiliary ODEs are presented in the Appendix $\ref{ApA}$.

\section{Market Model and Consumption Habit Formation}\label{sec1}
Given the probability space $(\Omega,\mathbb{F},\mathbb{P})$ with the filtration $\mathbb{F}=(\mathcal{F}_{t})_{0\leq{t}\leq{T}}$, which satisfies the usual conditions, we consider a financial market with one risk-free bond and one risky asset over a finite time horizon $[0, T]$. Without loss of generality, we assume that the bond process $S^0_t\equiv 1$, for all $t\in[0,T]$, following the standard change of num\'{e}raire.\\
\indent
The stock price $S_{t}$ is modeled as a diffusion process solving
\begin{equation}\label{stock}
dS_{t}=\mu_{t}S_{t}dt+\sigma_{S}S_{t}dW_t,\ \ \   0\leq{t}\leq{T},
\end{equation}
with $S_{0}=s>0$. Similar to \cite{BS} and \cite{Munk}, we assume that the drift process $\mu_{t}$ satisfies the Ornstein Uhlenbeck stochastic differential equation
\begin{equation}\label{mu}
d\mu_{t}=-\lambda(\mu_{t}-\bar{\mu})dt+\sigma_{\mu}dB_{t},\ \ \   0\leq{t}\leq{T}.
\end{equation}
Here, $(W_t)_{0\leq t\leq T}$ and $(B_t)_{0\leq t\leq T}$ are $(\mathcal{F}_{t})_{0\leq t\leq T}$-adapted Brownian motions with correlation coefficient $\rho\in[-1,1]$. The initial value of the drift process $\mu_{0}$ is assumed to be a $\mathcal{F}_{0}$-measurable Gaussian random variable, satisfying $\mu_{0}\thicksim\mathrm{N}(\eta_{0}, \theta_{0})$, which is independent of Brownian motions $W$ and $B$. We also assume that all coefficients $\sigma_{S}\geq 0,\lambda\geq 0,\bar{\mu},\sigma_{\mu}\geq{0}$ are constants.\\

\indent
At time $t\in{[0,\ T]}$, the investor chooses a consumption rate $c_{t}\geq{0}$ and decides the amounts $\pi_{t}$ of his wealth to invest in the risky asset, and the rest in the bond. The investor's total wealth process $X_{t}$ therefore follows the dynamics
\begin{equation}\label{wealth}
dX_{t}=(\pi_{t}\mu_{t}-c_{t})dt+\sigma_{S}\pi_{t}dW_{t},\ \ \  0\leq{t}\leq{T},
\end{equation}
with the initial wealth $X_{0}=x_{0}>0$.\\
\indent
In this paper, we adopt the notation $Z_{t}\triangleq Z(c)_t$ as \textquotedblleft Habit Formation\textquotedblright\ process or  \textquotedblleft the standard of living\textquotedblright\ process which describes the consumption habits level. We assume that the accumulative index $Z_t$ satisfies the recursive equation (see \cite{detemple92})
\begin{equation}\label{habit}
dZ_{t}=(\delta(t)c_{t}-\alpha(t)Z_{t})dt,\ \ \  0\leq{t}\leq{T},
\end{equation}
where \ $Z_{0}=z_{0}\geq{0}$ is called the \textit{initial habit}. Equivalently, we have
\begin{equation}
Z_t=z_0e^{-\int_0^t\alpha(u)du}+\int_0^t\delta(u)e^{-\int_u^t \alpha(v)dv}c_udu,\ \ \ 0\leq t\leq T.
\end{equation}
and it is the exponentially weighted average of the initial habit and the past consumption. Here, the discounted factors $\alpha(t)$ and $\delta(t)$ measure, respectively, the persistence of the past level and the intensity of consumption history. we assume $\alpha(t)$ and $\delta(t)$ to be nonnegative continuous functions.

In this paper, we only consider the case of \textit{addictive habits}, i.e., we require that the investor's current consumption strategies shall never fall below the standard of living level, 
\begin{equation}
c_t\geq Z_t,\ \ \ 0\leq t\leq T,\ \ \text{a.s.}.
\end{equation}

\section{Utility Maximization under Partial Observations}\label{sec2}

\subsection{Dynamic Programming Arguments}
From now on, we make the assumption that the investor can only observe the stock price process $S_{t}$ which is available to the public. The drift process $\mu_{t}$ and the information of Brownian motions $(W_{t})_{0\leq t\leq T}$ and $(B_{t})_{0\leq t\leq T}$ are unknown. The stochastic control problem under incomplete information will be modeled by requiring the investment strategy $(\pi_{t})_{0\leq{t}\leq{T}}$ and consumption policy $(c_{t})_{0\leq{t}\leq{T}}$ to be only adapted to the partial observation filtration $(\mathcal{F}_{t}^{S})_{0\leq t\leq T}$, which is strictly smaller than the background full information $\mathbb{F}=(\mathcal{F}_{t})_{0\leq t\leq T}$.\\
\indent
Applying the Kalman-Bucy filtering theorem, we can define the \textit{Innovation Process} by
\begin{equation}\label{innovaprocss}
d\hat{W}_{t}\triangleq \frac{1}{\sigma_{S}}\Big[(\mu_{t}-\hat{\mu}_{t})dt+\sigma_{S}dW_{t}\Big]=\frac{1}{\sigma_{S}}\Big(\frac{dS_{t}}{S_{t}}-\hat{\mu}_{t}dt\Big),\ \ \  0\leq{t}\leq{T},
\end{equation}
which is a Brownian motion under the partial observations filtration $\mathcal{F}^{S}_{t}$ and the process $\hat{\mu}_{t}=\mathbb{E}\Big[\mu_{t}\Big|\mathcal{F}^{S}_{t}\Big]$ is the conditional estimation of process $\mu_{t}$.\\
\indent
Moreover, by the Kalman-Bucy filtering theorem, the process $\hat{\mu}_{t}$ satisfies the linear SDE
\begin{equation}\label{newmu}
d\hat{\mu}_{t}= -\lambda(\hat{\mu}_{t}-\bar{\mu})dt+\Big(\frac{\hat{\Omega}_{t}+\sigma_{S}\sigma_{\mu}\rho}{\sigma_{S}}\Big)d\hat{W}_{t},\ \ \  0\leq{t}\leq{T},
\end{equation}
with $\hat{\mu}_{0}=\mathbb{E}\Big[\mu_{0}\Big|\mathcal{F}^{S}_{0}\Big]=\eta_{0}$. 
\indent
In addition, the conditional variance $\hat{\Omega}_{t}=\mathbb{E}\Big[(\mu_{t}-\hat{\mu}_{t})^{2}\Big|\mathcal{F}^{S}_{t}\Big]$ satisfies the deterministic Riccati ODE
\begin{equation}\label{oooga}
d\hat{\Omega}_{t}=\Big[-\frac{1}{\sigma_{S}^{2}}\hat{\Omega}_{t}^{2}+\Big(-\frac{2\sigma_{\mu}\rho}{\sigma_{S}}-2\lambda\Big)\hat{\Omega}_{t}+(1-\rho^{2})\sigma_{\mu}^{2}\Big]dt,\ \ \  0\leq{t}\leq{T},
\end{equation}
with $\hat{\Omega}(0)=\mathbb{E}\Big[(\mu_{0}-\eta)^{2}\Big|\mathcal{F}^{S}_{0}\Big]=\theta_{0}$, which has an explicit solution
\begin{equation}\label{newOmega}
\hat{\Omega}_{t}=\hat{\Omega}(t;\theta_{0})=\sqrt{k}\sigma_{S}\frac{k_{1}\exp(2(\frac{\sqrt{k}}{\sigma_{S}})t)+k_{2}}{k_{1}\exp(2(\frac{\sqrt{k}}{\sigma_{S}})t)-k_{2}}-(\lambda+\frac{\sigma_{\mu}\rho}{\sigma_{S}})\sigma_{S}^{2},\ \ 0\leq{t}\leq{T},
\end{equation}
and
\begin{equation}
\begin{split}
k & =\lambda^{2}\sigma^{2}_{S}{S}+2\sigma_{S}\sigma_{\mu}\lambda\rho+\sigma^{2}_{\mu},\\
k_{1} & =\sqrt{k}\sigma_{S}+(\lambda\sigma^{2}_{S}+\sigma_{S}\sigma_{\mu}\rho)+\theta_{0},\\
k_{2} & =-\sqrt{k}\sigma_{S}+(\lambda\sigma^{2}_{S}+\sigma_{S}\sigma_{\mu}\rho)+\theta_{0}.\nonumber
\end{split}
\end{equation}
\indent
It is easy to see that $\hat{\Omega}(t)$ converges monotonically to the value
\begin{equation}\label{stedst}
\theta^{\ast}=\sigma_{S}\sqrt{\lambda^{2}\sigma^{2}_{S}+2\sigma_{S}\sigma_{\mu}\lambda\rho+\sigma^{2}_{\mu}}-(\lambda\sigma^{2}_{S}+\sigma_{S}\sigma_{\mu}\rho)>0
\end{equation}
as time $t\rightarrow+\infty$, which we call as \textit{steady state learning} (see also \cite{Brennan}). This convergence property of $\hat{\Omega}(t)$ tells us the precision of the drift estimate goes from an initial condition to a steady state in the long time run, and after large time $T$, new return observations contribute to updating the estimated value of the state variable, but seldom reduce the variance of the estimation error. By the evolution of Riccati ODE $(\ref{oooga})$, we obtain that the monotone solution $\hat{\Omega}(t)$ on $(0,\infty)$ has the bounds
\begin{equation}
\min(\theta_{0},\theta^{\ast})\leq\hat{\Omega}(t)\leq \max(\theta_{0},\theta^{\ast}),\ \ \ \forall t\geq 0.
\end{equation}
Notice that $\min(\theta_{0},\theta^{\ast})$ and $\max(\theta_{0},\theta^{\ast})$ are independent of time $t$, therefore we can later provide some assumptions on the market coefficients independent of time to ensure that the verification results hold.

Under the observation filtration $(\mathcal{F}^{S}_{t})_{0\leq t\leq T}$, the stock price dynamics $(\ref{stock})$ can be rewritten by
\begin{equation}\label{newstocks}
dS_{t}=\hat{\mu}_{t}S_{t}dt+\sigma_{S}S_{t}d\hat{W}_{t},\ \ \  0\leq{t}\leq{T}.
\end{equation}

The habit formation process $Z_{t}$ still satisfies the ODE
\begin{equation}\label{newhabit}
dZ_{t}=(\delta(t)c_{t}-\alpha(t)Z_{t})dt,\ \ \  0\leq{t}\leq{T},\nonumber
\end{equation}
however, the consumption policy $c_t$ is now $\mathcal{F}^{S}_{t}$-progressively measurable.

Moreover, under $\mathcal{F}^{S}_{t}$-progressively measurable portfolio $\pi_t$ and consumption rate $c_t$, the wealth process dynamics $(\ref{wealth})$ can be rewritten as
\begin{equation}\label{newwealth}
dX_{t}=(\pi_{t}\hat{\mu}_{t}-c_{t})dt+\sigma_{S}\pi_{t}d\hat{W}_{t},\ \ \  0\leq{t}\leq{T}.
\end{equation}

The investment and consumption pair process $(\pi_t,c_t)$ is said to be in the \textbf{Admisisble Control Space}, denoted by $\mathcal{A}$, if it is $\mathcal{F}_t^S$-progressively measurable, and satisfies the integrability conditions
\begin{equation}
\int_0^T\pi_t^2 dt<+\infty, \ \ \text{a.s.}\ \ \text{and}\ \ \ \int_0^Tc_tdt<+\infty,\ \ \text{a.s.}
\end{equation}
with the addictive habit constraint that $c_t\geq Z_t$, $0\leq t\leq T$. Moreover, no bankruptcy is allowed, i.e., the investor's wealth remains nonnegative: $X_t\geq 0$, $0\leq t\leq T$.

Our goal is to maximize the consumption with habit formation and the terminal wealth with power utility preference under the partial observations filtration $\mathcal{F}^{S}_{t}$
\begin{equation}\label{problm}
v(x_{0},z_{0},\eta_{0},\theta_{0})=\underset{\pi,c\in\mathcal{A}}\sup\mathbb{E}\Big[\int_{0}^{T}\frac{(c_{s}-Z_{s})^{p}}{p}ds+\frac{(X_{T})^{p}}{p}\Big],
\end{equation}
where we take the risk aversion coefficient $p<1$ and $p\neq 0$.\\
\indent
In this paper, we aim to first solve the Dynamic Programming equation analytically and then perform a rigorous verification argument. Therefore, there is no need to either define the value function at later times or to prove the Dynamic Programming Principle involving some complicated measurable selection arguments. 

To this end, we look for a smooth function $\tilde{v}(t,x,z,\eta,\theta)$ defined on an appropriate domain such that the process
\begin{equation}
\int_{0}^{t}\frac{(c_{s}-Z_{s})^{p}}{p}ds+\tilde{v}(t,X_{t},Z_{t},\hat{\mu}_{t},\hat{\Omega}_{t}),\ \ \ \forall 0\leq t\leq T,\nonumber
\end{equation}
is a local supermartingale for each admissible control $(\pi_{t},c_{t})\in\mathcal{A}$ and a local martingale for the optimal feedback control $(\pi_{t}^{\ast},c_{t}^{\ast})\in\mathcal{A}$.

Since the conditional variance process $\hat{\Omega}_{t}=\hat{\Omega}(t,\theta_{0})$ is actually a deterministic function of time. We can therefore set the variable $\theta$ in the definition of $\hat{v}$ by a deterministic function $\theta=\theta(t,\theta_{0})$ depending on the parameter $\theta_{0}$ to reduce the dimension of the function $\tilde{v}$, i.e., the variable $\theta(t;\theta_{0})$ is absorbed by the variable $t$. Hence, we can define the function $V(t,x,z,\eta;\theta_{0})$ as
\begin{equation}
V(t,x,z,\eta;\theta_{0})\triangleq \tilde{v}(t,x,z,\eta,\theta(t,\theta_{0})),\nonumber
\end{equation}
and our target above can be simplified into finding a smooth enough function $V(t,x,z,\eta;\theta_{0})$ on some appropriate domain , denoted by $V(t,x,z,\eta)$ for simplicity. We expect that
\begin{equation}
\int_{0}^{t}\frac{(c_{s}-Z_{s})^{p}}{p}ds+V(t,X_{t},Z_{t},\hat{\mu}_{t}),\ \ \ \forall 0\leq t\leq T,\nonumber
\end{equation}
is a local supermartingale for each admissible control $(\pi_{t},c_{t})\in\mathcal{A}$ and a local martingale for the optimal feedback control $(\pi_{t}^{\ast},c_{t}^{\ast})\in\mathcal{A}$, for each fixed initial value $\hat{\Omega}(0)=\theta_{0}$.

By the definition of $V(t,x,z,\eta)$ and It\^{o}'s formula, we can formally derive the HJB equation as
\begin{equation}\label{HJBeqn}
\begin{split}
& V_{t}- \alpha(t)zV_{z}-\lambda(\eta-\bar{\mu})V_{\eta}+\frac{\Big(\hat{\Omega}(t)+\sigma_{S}\sigma_{\mu}\rho\Big)^{2}}{2\sigma^{2}_{S}}V_{\eta\eta}+\max_{c\in\mathcal{A}}\Big[-cV_{x}+c\delta(t)V_{z}\\
& + \frac{(c-z)^{p}}{p}\Big]+\max_{\pi\in\mathcal{A}}\Big[\pi\eta V_{x}+\frac{1}{2}\sigma^{2}_{S}\pi^{2}V_{xx}+V_{x\eta}\Big(\hat{\Omega}(t)+\sigma_{S}\sigma_{\mu}\rho\Big)\pi\Big]=0,
\end{split}
\end{equation}
with the terminal condition $V(T,x,z,\eta)=\frac{x^{p}}{p}$.\\

\subsection{The Decoupled Reduced Form Solution}\ \\
\indent
If $V(t,x,z,\eta)$ is smooth enough, the first order condition gives 
\begin{equation}\label{maxipolicy}
\begin{split}
\pi^{\ast}(t,x,z,\eta) & =\frac{-\eta V_{x}-\Big(\hat{\Omega}(t)+\sigma_{S}\sigma_{\mu}\rho\Big)\mathit{V_{x\eta}}}{\sigma^{2}_{S}V_{xx}},\\
c^{\ast}(t,x,z,\eta) & =z+\Big(V_{x}-\delta(t)V_{z}\Big)^{\frac{1}{p-1}}.
\end{split}
\end{equation}
which achieve the maximum over control policies $\pi$ and $c$ respectively. Moreover, we expect that the smooth solution $V(t,x,z,\eta)$ of the HJB equation at time $t=0$ is actually the value function, i.e., $V(0,x_{0},z_{0},\eta_{0};\theta_{0})=v(x_{0},z_{0},\eta_{0},\theta_{0})$. Due to the homogeneity property of the power utility function and the linearity of dynamics $(\ref{newwealth})$ and $(\ref{newhabit})$ for $X_{t}$ and $Z_{t}$ respectively, it's easy to see that if $V(t,x,z,\eta)$ is finite, then it is homogeneous in $(x,z)$ with degree $p$, i.e., for any $x>0$, $z\geq 0$ and the positive constant $k$, we have $V(t,kx,kz,\eta)=k^{p}V(t,x,z,\eta)$. It therefore makes sense for us to seek the value function of the form:
\begin{equation}
V(t,x,z,\eta)=\frac{\Big[(x-m(t,\eta)z)\Big]^{p}}{p}M(t,\eta)\nonumber
\end{equation}
for some test functions $m(t,\eta)$ and $M(t,\eta)$ to be determined. By the virtue of $V(T)=\frac{x^{p}}{p}$, we will require $M(T,\eta)=1$ and $m(T,\eta)=0$.\\
\indent
After the direct substitutions and dividing the equation on both sides by $(x-m(t,\eta)z)^{p}$, the HJB equation $(\ref{HJBeqn})$ becomes
\begin{equation}\label{hjb333}
\begin{split}
& \frac{\Big[f(t,m)z+\lambda(\eta-\bar{\mu})m_{\eta}-\frac{1}{2\sigma^{2}_{S}}(\hat{\Omega}(t)+\sigma_{S}\sigma_{\mu}\rho)^{2}m_{\eta\eta}+\frac{\eta}{\sigma^{2}_{S}}(\hat{\Omega}(t)+\sigma_{S}\sigma_{\mu}\rho)m_{\eta}\Big]z}{x-m(t,\eta)z}
M
\\
& +\frac{1}{p}M_{t}-\frac{\lambda(\eta-\bar{\mu})}{p}M_{\eta}+\frac{(\hat{\Omega}(t)+\sigma_{S}\sigma_{\mu}\rho)^{2}}{2p\sigma^{2}_{S}}M_{\eta\eta}-\frac{\eta(\hat{\Omega}(t)+\sigma_{S}\sigma_{S}\rho)}{(p-1)\sigma^{2}_{S}}M_{\eta} \\
& -\frac{\eta^{2}}{2(p-1)\sigma^{2}_{S}}M-\frac{(\hat{\Omega}(t)+\sigma_{S}\sigma_{\mu}\rho)^{2}}{2(p-1)\sigma^{2}_{S}}\frac{M_{\eta}^{2}}{M}-\frac{p-1}{p}\Big(1+\delta(t)m(t,\eta)\Big)^{\frac{p}{p-1}}M^{\frac{p}{p-1}}=0.
\end{split}
\end{equation}
where we set
\begin{equation}
f(t,m)=-m_{t}+\alpha(t)m-(1+\delta(t)m).\nonumber
\end{equation}
\indent
Since the Equation $(\ref{hjb333})$ above holds for all values of $x$ and $z$, we must have
\begin{equation}\label{ffunc}
f(t,m)+\lambda(\eta-\bar{\mu})m_{\eta}-\frac{1}{2\sigma_S^2}(\hat{\Omega}(t)+\sigma_S\sigma_{\mu}\rho)^2m_{\eta\eta}+\frac{\eta}{\sigma_S^2}(\hat{\Omega}(t)+\sigma_S\sigma_{\mu}\rho)m_{\eta}=0,\nonumber
\end{equation}
with $m(T, \eta)=0$. We propose to set the unknown priori function $m(t,\eta)=m(t)$ as a deterministic function in time $t$ with the terminal condition $m(T)=0$ and hence $f(t,m)=0$, which is equivalent to
\begin{equation}\label{www}
m(t)=\int_{t}^{T}\exp\Big(\int_{t}^{s}(\delta(v)-\alpha(v))dv\Big)ds.\ \ \ 0\leq{t}\leq{T}.
\end{equation}
Clearly, $m(t)$ solves the above equation $(\ref{ffunc})$.

By substituting $m(t)$ into the equation $(\ref{hjb333})$, it is simplified as
\begin{equation}\label{equ4}
\begin{split}
& M_{t}+\frac{p\eta^{2}}{2(1-p)\sigma_{S}^{2}}M+\frac{\Big(\hat{\Omega}(t)+\sigma_{S}\sigma_{\mu}\rho\Big)^{2}}{2\sigma^{2}_{S}}M_{\eta\eta}+(1-p)\Big(1+\delta(t)m(t)\Big)^{\frac{p}{p-1}}M^{\frac{p}{p-1}}
\\
& +\Big[-\lambda(\eta-\bar{\mu})+\frac{\eta(\hat{\Omega}(t)+\sigma_{S}\sigma_{\mu}\rho)p}{(1-p)\sigma^{2}_{S}}\big]M_{\eta}+\frac{\Big(\hat{\Omega}(t)+\sigma_{S}\sigma_{\mu}\rho\Big)^{2}p}{2(1-p)\sigma^{2}_{S}}\frac{M_{\eta}^{2}}{M}=0.
\end{split}
\end{equation}
\indent
Now in order to solve the above nonlinear PDE $(\ref{equ4})$, we can set the power transform as
\begin{equation}\label{MN}
M(t,\eta)=N(t,\eta)^{1-p}
\end{equation}
This idea of power transform was first introduced in \cite{MR1807876}.\\
\indent
And the nonlinear PDE $(\ref{equ4})$ for $M(t,\eta)$ reduces to the linear parabolic PDE for $N(t,\eta)$ as:

\begin{equation}\label{equ5}
\begin{split}
& N_{t}+\frac{p\eta^{2}}{2(1-p)^{2}\sigma^{2}_{S}}N(t,\eta)+\frac{\Big(\hat{\Omega}(t)+\sigma_{S}\sigma_{\mu}\rho\Big)^{2}}{2\sigma^{2}_{S}}N_{\eta\eta}+\Big(1+\delta(t)m(t)\Big)^{\frac{p}{p-1}}
\\
& +\Big[-\lambda(\eta-\bar{\mu})+\frac{\eta\Big(\hat{\Omega}(t)+\sigma_{S}\sigma_{\mu}\rho\Big)p}{(1-p)\sigma^{2}_{S}}\Big]N_{\eta}(t,\eta)=0 
\end{split}
\end{equation}
with $N(T,\eta)=1$.\\
\indent
For the above linear PDE $(\ref{equ5})$ of $N(t,\eta)$, we can further solve it explicitly by
\begin{equation}\label{eqnN}
\begin{split}
N(t,\eta)= & \int_{t}^{T}\Big(1+\delta(s)m(s)\Big)^{\frac{p}{p-1}}\exp\Big(A(t,s)\eta^{2}+B(t,s)\eta+C(t,s)\Big)ds\\
& +\exp\Big(A(t,T)\eta^{2}+B(t,T)\eta+C(t,T)\Big),
\end{split}
\end{equation}
where for $0\leq{t}\leq{s}\leq{T}$, $A(t,s)$, $B(t,s)$ and $C(t,s)$ satisfy the following ODEs:
\begin{equation}\label{odeA}
A_{t}(t,s)+\frac{p}{2(1-p)^{2}\sigma^{2}_{S}}+2\Big[-\lambda+\frac{p(\hat{\Omega}(t)+\sigma_{S}\sigma_{\mu}\rho)}{\sigma^{2}_{S}(1-p)}\Big]A(t,s)+\frac{2(\hat{\Omega}(t)+\sigma_{S}\sigma_{\mu}\rho)^{2}}{\sigma^{2}_{S}}A^{2}(t,s)=0;
\end{equation}
\begin{equation}\label{odeB}
B_{t}(t,s)+\Big[-\lambda+\frac{p(\hat{\Omega}(t)+\sigma_{S}\sigma_{\mu}\rho)}{\sigma^{2}_{S}(1-p)}\Big]B(t,s)+2\lambda\bar{\mu}A(t,s)+\frac{2(\hat{\Omega}(t)+\sigma_{S}\sigma_{\mu}\rho)^{2}}{\sigma^{2}_{S}}A(t,s)B(t,s)=0;\ \
\end{equation}
\begin{equation}\label{odeC}
C_{t}(t,s)+\lambda\bar{\mu}B(t,s)+\frac{\Big(\hat{\Omega}(t)+\sigma_{S}\sigma_{\mu}\rho\Big)^{2}}{2\sigma^{2}_{S}}\Big(B^{2}(t,s)+2A(t,s)\Big)=0;\ \ \ \ \ \ \ \ \ \ \ \ \ \ \ \ \
\end{equation}
with terminal conditions:   $A(s,s)=B(s,s)=C(s,s)=0$.\\

We remark that the above ODEs are similar to the ODEs obtained by \cite{BS} for terminal wealth optimization problem with partial observations in which an insightful observation is made that we can solve the above $3$ ODEs with time $t$ dependent coefficients by solving the following $5$ auxiliary ODEs with constant coefficients, see section $4$ of \cite{BS} for the detail proof.
\begin{lemma}\label{brene}
For $0\leq t\leq s\leq T$, consider the following auxiliary ODEs for $a(t,s)$, $b(t,s)$, $c(t,s)$, $f(t,s)$ and $g(t,s)$:
\begin{equation}\label{eqna}
a_{t}=-\frac{2(1-p+p\rho^{2})}{1-p}\sigma^{2}_{\mu}a^{2}+\Big(2\lambda-\frac{2p\rho\sigma_{\mu}}{(1-p)\sigma_{S}}\Big)a-\frac{p}{2(1-p)\sigma^{2}_{S}},
\end{equation}
\begin{equation}\label{eqnb}
b_{t}=-\frac{2(1-p+p\rho^{2})}{1-p}\sigma^{2}_{\mu}ab-2\lambda\bar{\mu}a+\Big(\lambda-\frac{p\rho\sigma_{\mu}}{(1-p)\sigma_{S}}\Big)b,
\end{equation}
\begin{equation}\label{eqnc}
c_{t}=-\sigma^{2}_{\mu}a-\frac{(1-p+p\rho^{2})\sigma^{2}_{\mu}}{2(1-p)}b^{2}(t)-\lambda\bar{\mu}b,
\end{equation}
\begin{equation}\label{eqnL}
f_{t}=-2(1-\rho^{2})\sigma^{2}_{\mu}f^{2}+2\frac{\lambda\sigma_{S}+\rho\sigma_{\mu}}{\sigma_{S}}f+\frac{1}{2\sigma^{2}_{S}},
\end{equation}
\begin{equation}\label{eqnF}
g_{t}=\sigma^{2}_{\mu}(1-\rho^{2})(f-a),
\end{equation}
with the terminal conditions $a(s,s)=b(s,s)=c(s,s)=f(s,s)=g(s,s)=0$. If we adopt the convention $\frac{0}{0}=0$ and consider the functions defined by
\begin{equation}
\tilde{A}(t,s)\triangleq \frac{a(t,s)}{(1-p)\Big(1-2a(t,s)\hat{\Omega}(t)\Big)},\nonumber
\end{equation}
\begin{equation}
\tilde{B}(t,s)\triangleq \frac{b(t,s)}{(1-p)\Big(1-2a(t,s)\hat{\Omega}(t)\Big)},\nonumber
\end{equation}
\begin{equation}
\begin{split}
\tilde{C}(t,s) \triangleq & \frac{1}{1-p}\Big[c(t,s)+\frac{\hat{\Omega}(t)}{\Big(1-2a(t,s)\hat{\Omega}(t)\Big)}b^{2}(t,s)-\frac{1-p}{2}\log\Big(1-2a(t,s)\hat{\Omega}(t)\Big)\\
& -\frac{p}{2}\log\Big(1-2f(t,s)\hat{\Omega}(t)\Big)-pg(t,s)\Big],\nonumber
\end{split}
\end{equation}
the following equivalence results hold
\begin{equation}\label{eeeq}
A(t,s)=\tilde{A}(t,s),\ B(t,s)=\tilde{B}(t,s),\ C(t,s)=\tilde{C}(t,s),\ \ \ 0\leq t\leq s\leq T.
\end{equation}
\end{lemma}

We can further find fully explicit forms for $a(t,s)$, $b(t,s)$, $c(t,s)$, $f(t,s)$ and $g(t,s)$. We list all four different cases of explicit solutions in the Appendix $\ref{ApA}$ depending on the risk aversion coefficient $p$ and the market coefficients $\sigma_{S}$, $\sigma_{\mu}$, $\lambda$ and $\rho$. By simple substitutions, we can therefore solve the ODEs $(\ref{odeA})$, $(\ref{odeB})$, $(\ref{odeC})$ for $A(t,s)$, $B(t,s)$ and $C(t,s)$ fully explicitly.

For $t\in[0,T]$, $\eta\in(-\infty,+\infty)$, we can define the \textit{effective domain} for the pair $(x,z)$ by
\begin{equation}
(x,z)\in\ \mathbb{D}_{t}=\{(x',z')\in \ (0,+\infty)\times[0,+\infty);\ x'\geq m(t)z' \},\ 0\leq{t}\leq{T}.
\end{equation}
The function
\begin{equation}\label{solut}
\begin{split}
V(t,x,z,\eta)= & \Big[\int_{t}^{T}\Big(1+\delta(s)m(s)\Big)^{\frac{p}{p-1}}\exp\Big(A(t,s)\eta^{2}+B(t,s)\eta+C(t,s)\Big)ds\\
& +\exp\Big(A(t,T)\eta^{2}+B(t,T)\eta+C(t,T)\Big)\Big]^{1-p}\frac{[(x-m(t)z)]^{p}}{p}
\end{split}
\end{equation}
is well defined on $[0,T]\times \mathbb{D}_{t}\times \mathbb{R}$ and it's the classical solution of the HJB equation $(\ref{HJBeqn})$, where $m(t)=\int_{t}^{T}\exp(\int_{t}^{s}(\delta(v)-\alpha(v))dv)ds$, and $A(t,s)$, $B(t,s)$, $C(t,s)$ are solutions of ODEs $(\ref{odeA})$, $(\ref{odeB})$, $(\ref{odeC})$. 

\begin{remark}
In our main result below, we want to verify that the above classical solution $V(t,x,z,\eta)$ at time $t=0$ equals the primal value function defined in $(\ref{problm})$, i.e., $V(0,x_{0},z_{0},\eta_{0};\theta_{0})=v(x_{0},z_{0},\eta_{0},\theta_{0})$. However, the effective domain of $V(t,x,z,\eta)$ motivates some constraints on the optimal wealth process $X_{t}^{\ast}$ and habit formation process $Z_{t}^{\ast}$. To see this, we have $V(t,x,z,\eta)=-\infty$ when $x<m(t)z$, which mandates that $X_{t}^{\ast}\geq m(t)Z_{t}^{\ast}$ for each $t\in[0,T]$ to ensure that the process $V(t,X_{t}^{\ast},Z_{t}^{\ast},\hat{\mu}_{t})$ is well defined. In particular, when $t=0$, we have to enforce the initial wealth-habit budget constraint that $x_{0}\geq m(0)z_{0}$.
\end{remark}

\begin{assumption}\label{ass111}
The risk aversion constant $1-p$ of the utility function satisfies
\begin{equation}
p<0.\nonumber
\end{equation}
\end{assumption}

\begin{assumption}\label{ass222}
The risk aversion constant $1-p$ of the utility function satisfies
\begin{equation}
0<p<1.\nonumber
\end{equation}
And the explicit functions $a(t,s)$, $b(t,s)$, $c(t,s)$, $f(t,s)$ and $g(t,s)$ solved in Lemma $\ref{brene}$ are all bounded and $1-a(t,s)\hat{\Omega}(t)\neq 0$ on $0\leq t\leq s\leq T$. (See Appendix \ref{ApA})

Moreover, we assume that
\begin{equation}
\frac{2p^2+p}{(1-p)^2}<\frac{\lambda^2\sigma_S^4}{4(\Theta+\sigma_S\sigma_{\mu}\rho)^2},
\end{equation}
where $\Theta\triangleq \max\{\theta_0, \theta^{\ast}\}$ and $\theta^{\ast}$ is defined in $(\ref{stedst})$. The upper bound $\bar{K}_1$ of $A(t,s)$ on $0\leq t\leq s\leq T$ satisfies
\begin{equation}\label{newc}
\bar{K}_1 <\frac{\lambda \sigma_S^2}{4(\Theta+\sigma_S\sigma_{\mu}\rho)^2}.
\end{equation}
\end{assumption}

\subsection{The Main Result}
\begin{theorem}[The Verification Theorem]\label{mainmain}
If the initial wealth-habit budget constraint $x_0>m(0)z_0$ holds, under either Assumption $\ref{ass111}$ or Assumption $\ref{ass222}$, the solution $(\ref{solut})$ of HJB equation equals the value function defined in $(\ref{problm})$:
\begin{equation}\label{veri}
V(0,x_{0},z_{0},\eta_{0};\theta_{0})=v(x_{0},z_{0},\eta_{0},\theta_{0}).
\end{equation}

Moreover, the optimal investment policy $\pi^{\ast}_{t}$ and optimal consumption policy $c^{\ast}_{t}$ are given in the feedback form: $\pi^{\ast}_{t}=\pi^{\ast}(t,X^{\ast}_{t},Z^{\ast}_{t},\hat{\mu}_{t})$ and $c^{\ast}_{t}=c^{\ast}(t,X^{\ast}_{t},Z^{\ast}_{t},\hat{\mu}_{t})$, $0\leq{t}\leq{T}$, where the function $\pi^{\ast}(t,x,z,\eta):[0,T]\times\mathbb{D}_{t}\times\mathbb{R}\longrightarrow\mathbb{R}$ is defined by:
\begin{equation}\label{optpi}
\pi^{\ast}(t,x,z,\eta)=\Big[\frac{\eta}{(1-p)\sigma^{2}_{S}}+\frac{\Big(\hat{\Omega}(t)+\sigma_{S}\sigma_{\mu}\rho\Big)}{\sigma^{2}_{S}}\frac{N_{\eta}(t,\eta)}{N(t,\eta)}\Big](x-m(t)z),\ 0\leq{t}\leq{T}.
\end{equation}
$c^{\ast}(t,x,z,\eta):[0,T]\times\mathbb{D}_{t}\times\mathbb{R}\longrightarrow\mathbb{R}_{+}$
is defined by:
\begin{equation}\label{optc}
c^{\ast}(t,x,z,\eta)=z+\frac{(x-m(t)z)}{\big(1+\delta(t)m(t)\big)^{\frac{1}{1-p}}N(t,\eta)},\ \ 0\leq{t}\leq{T}.
\end{equation}
\indent
The corresponding optimal wealth process $X^{\ast}_{t}$, for $0\leq{t}\leq{T}$, can also be found explicitly by
\begin{equation}\label{optiwealth}
X_{t}^{\ast}=(x_{0}-m(0)z_{0})\frac{N(t,\hat{\mu}_{t})}{N(0,\eta)}\exp\Big(\int_{0}^{t}\frac{(\hat{\mu}_{u})^{2}}{2(1-p)\sigma^{2}_{S}}du+\int_{0}^{t}\frac{\hat{\mu}_{u}}{(1-p)\sigma_{S}}d\hat{W}_{u}\Big)+m(t)Z^{\ast}_{t},
\end{equation}
where $m(t)$ and $N(t,\eta)$ are defined in $(\ref{www})$ and $(\ref{eqnN})$ respectively.
\end{theorem}

The complicated structure of feedback forms of optimal investments and consumption policies is the consequence of the time non-separability of the instantaneous utility with habit formation. We can see that  the portfolio/wealth ratio $\frac{\pi^{\ast}}{X^{\ast}}$ and consumption/wealth ratio $\frac{c^{\ast}}{X^{\ast}}$ are depending on the habit-formation/wealth ratio $\frac{Z^{\ast}}{X^{\ast}}$:
\begin{equation}
\frac{\pi^{\ast}_t}{X^{\ast}_t}=\Big[\frac{\hat{\mu}_{t}}{(1-p)\sigma^{2}_{S}}+\frac{\Big(\hat{\Omega}(t)+\sigma_{S}\sigma_{\mu}\rho\Big)}{\sigma^{2}_{S}}\frac{N_{\eta}(t,\hat{\mu}_{t})}{N(t,\hat{\mu}_{t})}\Big](1-m(t)\frac{Z^{\ast}_t}{X^{\ast}_t}),\nonumber
\end{equation}
and
\begin{equation}
\frac{c^{\ast}_t}{X_t^{\ast}}=\frac{1}{\big(1+\delta(t)m(t)\big)^{\frac{1}{1-p}}N(t,\hat{\mu}_{t})}+\Big(1-\frac{m(t)}{\big(1+\delta(t)m(t)\big)^{\frac{1}{1-p}}N(t,\hat{\mu}_{t})}\Big)\frac{Z_t^{\ast}}{X_t^{\ast}},\nonumber
\end{equation}
\indent
Moreover, based on the explicit structures, we can easily provide some quantitative analysis on the value function and the optimal strategies.
\begin{corollary}
We have the following basic properties:
\begin{itemize}
\item[(1)] The value function $V(t,x,z,\eta)$ is concave in $x$ and $z$. And $V(t,x,z,\eta)$ is increasing in $x$ and decreasing in $z$.

\item[(2)] The feedback form $\pi^{\ast}(t,x,z,\eta)$ is linear in $x$ and $z$. And if $\frac{\eta}{(1-p)\sigma^{2}_{S}}+\frac{\Big(\hat{\Omega}(t)+\sigma_{S}\sigma_{\mu}\rho\Big)}{\sigma^{2}_{S}}\frac{N_{\eta}(t,\eta)}{N(t,\eta)}>0$, $\pi^{\ast}$ is increasing in $x$ and decreasing in $z$ and similarly, if $\frac{\eta}{(1-p)\sigma^{2}_{S}}+\frac{\Big(\hat{\Omega}(t)+\sigma_{S}\sigma_{\mu}\rho\Big)}{\sigma^{2}_{S}}\frac{N_{\eta}(t,\eta)}{N(t,\eta)}<0$, $\pi^{\ast}$ is decreasing in $x$ and increasing in $z$. 

\item[(3)] The feedback form $c^{\ast}(t,x,z,\eta)$ is linear in $x$ and $z$ and $c^{\ast}$ is increasing in $x$. If $1-\frac{m(t)}{(1+\delta(t)m(t))^{\frac{1}{1-p}}N(t,\eta)}>0$, $c^{\ast}$ is increasing in $z$ and if $1-\frac{m(t)}{(1+\delta(t)m(t))^{\frac{1}{1-p}}N(t,\eta)}<0$, $c^{\ast}$ is decreasing in $z$.
\end{itemize}
\end{corollary}

\section{Proof of The Verification Theorem}\label{sec3}
We first show that the consumption constraint $c_{t}\geq Z_{t}$ implies the constraint on the controlled wealth process by the following lemma
\begin{lemma}\label{yu}
The admissible space $\mathcal{A}$ is not empty if and only if the initial budget constraint $x_{0}\geq m(0)z_{0}$ is fulfilled. Moreover, for each pair of investment and consumption policy $(\pi,c)\in\mathcal{A}$, the controlled wealth process $X^{\pi,c}_{t}$ satisfies the constraint
\begin{equation}\label{kyu}
X_{t}^{\pi, c}\geq m(t)Z_{t},\ \ \ 0\leq{t}\leq{T},
\end{equation}
where the deterministic function $m(t)$ is defined in $(\ref{www})$ and refers to the cost of subsistence consumption per unit of standard of living at time t.
\end{lemma}
\begin{proof}.
Let's first assume that $x_{0}\geq m(0)z_{0}$, we can always take $\pi_{t}\equiv 0$, and $c_{t}=z_{0}e^{\int_{0}^{t}(\delta(v)-\alpha(v))dv}$ for $t\in[0,T]$. It is easy to verify $X_{t}^{\pi,c}\geq 0$ and $c_{t}\equiv Z_{t}$ so that $(\pi,c)\in\mathcal{A}$, and hence $\mathcal{A}$ is not empty.\\
\indent
On the other hand, starting from $t=0$ with the wealth $x_{0}$ and the standard of living $z_{0}$, the addictive habits constraint $c_{t}\geq Z_{t},\ 0\leq{t}\leq{T}$ implies that the consumption must always exceed the \textit{subsistence consumption} $\bar{c}_{t}=Z(t;\bar{c}_{t})$ which satisfies
\begin{equation}\label{newappend2}
d\bar{c}_{t}=(\delta(t)-\alpha(t))\bar{c}_{t}dt,\ \ \bar{c}_{0}=z_{0},\ \ 0\leq{t}\leq{T}.
\end{equation}
Indeed, since $Z_{t}$ satisfies $dZ_{t}=(\delta_{t}c_{t}-\alpha_{t}Z_{t})dt$ with $Z_{0}=z\geq 0$, the constraint $c_{t}\geq Z_{t}$ implies that
\begin{equation}\label{newappend1}
dZ_{t}\geq (\delta_{t}Z_{t}-\alpha_{t}Z_{t})dt,\ \ \ Z_{0}=z_{0}.
\end{equation}
\indent
By $(\ref{newappend2})$ and $(\ref{newappend1})$, one can get
\begin{equation}
d(Z_{t}-\bar{c}_{t})\geq (\delta_{t}-\alpha_{t})(Z_{t}-\bar{c}_{t})dt,\ \ \ Z_{0}-\bar{c}_{0}=0,\nonumber
\end{equation}
from which we can derive that
\begin{equation}
e^{\int_{0}^{t}(\delta_{s}-\alpha_{s})ds}(Z_{t}-\bar{c}_{t})\geq 0, \ \ \ 0\leq t\leq T.
\end{equation}
It follows that $c_{t}\geq \bar{c}_{t}$, which is equivalent to
\begin{equation}\label{consccc}
c_{t}\geq z_{0} e^{\int_{0}^{t}(\delta(v)-\alpha(v))dv},\ \ 0\leq t\leq T.
\end{equation}
\indent
Define the exponential local martingale
\begin{equation}\label{stockdensityH}
\widetilde{H}_{t}=\exp\Big(-\int_{0}^{t}\frac{\hat{\mu}_{v}}{\sigma_{S}}d\hat{W}_{v}-\frac{1}{2}\int_{0}^{t}\frac{\hat{\mu}_{v}^{2}}{\sigma^{2}_{S}}dv\Big),\ \ 0\leq{t}\leq{T}.
\end{equation}
Since $\hat{\mu}_{t}$ follows the dynamics $(\ref{newmu})$, which is
\begin{equation}
\hat{\mu}_{t}=e^{-t\lambda}\eta+\bar{\mu}(1-e^{-t\lambda})+\int_{0}^{t}e^{\lambda(u-t)}\frac{\Big(\hat{\Omega}(u)+\sigma_{S}\sigma_{\mu}\rho\Big)}{\sigma_{S}}d\hat{W}_{u}.\nonumber
\end{equation}
Similar to the proof of Corollary $3.5.14$ and Corollary $3.5.16$ in \textit{Karatzas and Shreve} \cite{MR917065}, Bene\v{s}' condition implies $\widetilde{H}$ is a true martingale with respect to $(\Omega, \mathcal{F}^{S},\mathbb{P})$,\\
\indent
Now define the probability measure $\widetilde{\mathbb{P}}$ as
\begin{equation}
\frac{d\widetilde{\mathbb{P}}}{d\mathbb{P}}=\widetilde{H}_{T},\nonumber
\end{equation}
Girsanov theorem states that
\begin{equation}
\widetilde{W}_{t}\triangleq \hat{W}_{t}+\int_{0}^{t}\frac{\hat{\mu}_{v}}{\sigma_{S}}dv,\ \ \ 0\leq{t}\leq T\nonumber
\end{equation}
is a Brownian Motion under $(\widetilde{\mathbb{P}}, (\mathcal{F}^{S}_{t})_{0\leq{t}\leq T})$.\\
\indent
We can rewrite the wealth process dynamics by
\begin{equation}
X_{T}+\int_{0}^{T}c_{v}dv=x+\int_{0}^{T}\pi_{v}\sigma_{S}d\widetilde{W}_{v},\nonumber
\end{equation}
Since we have $X_{T}\geq{0}$, it's easy to see that $\int_{0}^{t}\pi_{v}\sigma_{S}d\widetilde{W}_{v}$ is a supermartingale under $(\Omega, \mathbb{F}^{S},\widetilde{\mathbb{P}})$. By taking the expectation under $\widetilde{\mathbb{P}}$, we have:
\begin{equation}
x_{0}\geq{}\widetilde{\mathbb{E}}\Big[\int_{0}^{T}c_{v}dv\Big].\nonumber
\end{equation}
Follow the inequality $(\ref{consccc})$, we will further have
\begin{equation}
x_{0}\geq z_{0}\widetilde{\mathbb{E}}\Big[\int_{0}^{T}\exp\Big(\int_{0}^{v}(\delta(u)-\alpha(u))du\Big)dv\Big].\nonumber
\end{equation}
Since $\delta(t)$ and $\alpha(t)$ are deterministic functions, we obtain that $x_{0}\geq m(0)z_{0}$.\\
\indent
In general, for $\forall t\in[0,T]$, follow the same procedure, we can take conditional expectation under filtration $\mathcal{F}^{S}_{t}$, and get
\begin{equation}
X_{t}\geq Z_{t}\widetilde{\mathbb{E}}\Big[\int_{t}^{T}\exp\Big(\int_{t}^{v}(\delta(u)-\alpha(u))du\Big)dv\Big|\mathcal{F}_{t}^{S}\Big].\nonumber
\end{equation}
Again since $\delta(t),\alpha(t)$ are deterministic, we get $X_{t}\geq m(t)Z_{t},\ \ 0\leq{t}\leq{T}$.
\end{proof}
\begin{remark}
The constraint on the controlled wealth process $X_{t}$ and the habit formation process $Z_{t}$ agrees with the effective domain $\{(x,z)\in(0,\infty)\times[0,\infty):x\geq m(t)z\}$ of the HJB equation for the values of $x$ and $z$. Aside from the consequence that the process $V(t,X_{t},Z_{t},\hat{\mu}_{t})$ is therefore well defined, it plays a critical role in our following proof of the verification theorem.
\end{remark}

We proceed to show the proof of the Theorem $\ref{mainmain}$. 

\subsection{The Case $\mathit{p<0}$}
\begin{proof}[(PROOF OF THEOREM $\ref{mainmain}$)]\ \\
\indent
For any pair of admissible control $(\pi_{t},c_{t})\in\mathcal{A}$, It\^{o}'s lemma gives
\begin{equation}\label{preeqn} 
d\Big[V(t,X_{t},Z_{t},\hat{\mu}_{t})\Big]=\Big[\mathcal{G}^{\pi_{t},c_{t}}V(t,X_{t},Z_{t},\hat{\mu}_{t})\Big]dt+\Big[V_{x}\sigma_{S}\pi_{t}+V_{\eta}\frac{\Big(\hat{\Omega}(t)+\sigma_{S}\sigma_{\mu}\rho\Big)}{\sigma_{S}}\Big]d\hat{W}_{t},
\end{equation}
where we define the process $\mathcal{G}^{\pi_{t},c_{t}}V(t,X_{t},Z_{t},\hat{\mu}_{t})$ by
\begin{equation}
\begin{split}
& \mathcal{G}^{\pi_{t},c_{t}}V(t,X_{t},Z_{t},\hat{\mu}_{t}) = V_{t}- \alpha(t)Z_{t}V_{z}-\lambda(\hat{\mu}_{t}-\bar{\mu})V_{\eta}+\frac{\Big(\hat{\Omega}(t)+\sigma_{S}\sigma_{\mu}\rho\Big)^{2}}{2\sigma^{2}_{S}}V_{\eta\eta}\\
& -c_{t}V_{x}+c_{t}\delta(t)V_{z}+\frac{(c_{t}-Z_{t})^{p}}{p}+\pi_{t}\hat{\mu}_{t}V_{x}+\frac{1}{2}\sigma^{2}_{S}\pi_{t}^{2}V_{xx}+V_{x\eta}\Big(\hat{\Omega}(t)+\sigma_{S}\sigma_{\mu}\rho\Big)\pi_{t}.\nonumber
\end{split}
\end{equation}
\indent
For any localizing sequence $\tau_{n}$, by integrating the equation $(\ref{preeqn})$ on $[0, \tau_{n}\wedge T]$ and taking the expectation, it follows that
\begin{equation}\label{inqV}
V(0,x_{0},z_{0},\eta_{0})\geq{\mathbb{E}\Big[\int_{0}^{\tau_{n}\wedge T}\frac{(c_{s}-Z_{s})^{p}}{p}ds\Big]+\mathbb{E}\Big[V(\tau_{n}\wedge T,X_{\tau_{n}\wedge T},Z_{\tau_{n}\wedge T},\hat{\mu}_{\tau_{n}\wedge T})\Big]}.
\end{equation}
\indent
Similar to the idea by \cite{Sirbu}, we consider the fixed pair of control $(\pi_{t},c_{t})\in\mathcal{A}=\mathcal{A}_{x_{0}}$, where we denote $\mathcal{A}_{x_{0}}$ as the admissible space with initial endowment $x_{0}$. For $\forall \epsilon >0$, it is clear that $\mathcal{A}_{x_{0}}\subseteq\mathcal{A}_{x_{0}+\epsilon}$, and $(\pi_{t},c_{t})\in \mathcal{A}_{x_{0}+\epsilon}$. Also it is easy to see that $X^{x_{0}+\epsilon}_{t}=X^{x_{0}}_{t}+\epsilon=X_{t}+\epsilon,\ \ \ 0\leq{t}\leq{T}$. Since the process $Z_{t}$ is defined using this consumption policy $c_{t}$, under the probability measure $\mathbb{P}_{x_{0},z_{0},\eta_{0}}$, we can obtain
\begin{equation}\label{equnV}
V(0,x_{0}+\epsilon,z_{0},\eta_{0}) \geq\mathbb{E}\Big[\int_{0}^{\tau_{n}\wedge T}\frac{(c_{s}-Z_{s})^{p}}{p}ds\Big]+\mathbb{E}\Big[V(\tau_{n}\wedge T,X_{\tau_{n}\wedge T}+\epsilon,Z_{\tau_{n}\wedge T},\hat{\mu}_{\tau_{n}\wedge T})\Big].
\end{equation}
\indent
Monotone Convergence Theorem first gives that
\begin{equation}\label{mct}
\lim_{n\rightarrow +\infty}\mathbb{E}\Big[\int_{0}^{\tau_{n}\wedge T}\frac{(c_{s}-Z_{s})^{p}}{p}ds\Big]=\mathbb{E}\Big[\int_{0}^{T}\frac{(c_{s}-Z_{s})^{p}}{p}ds\Big].
\end{equation}
\indent
For simplicity, let's denote $Y_{t}=\Big(X_{t}-m(t)Z_{t}\Big)$. The definition $(\ref{solut})$ implies that:
\begin{equation}
V(\tau_{n}\wedge T,X_{\tau_{n}\wedge T}+\epsilon,Z_{\tau_{n}\wedge T},\hat{\mu}_{\tau_{n}\wedge T})=\frac{1}{p}(Y_{\tau_{n}\wedge T}+\epsilon)^{p}\mathit{N}^{1-p}_{\tau_{n}\wedge T}.\nonumber
\end{equation}
\indent
Lemma $\ref{yu}$ gives $X_{t}\geq W(t)Z_{t}$ for $0\leq t\leq T$ under any admissible control $(\pi_{t},c_{t})$, we get that $Y_{\tau_{n}\wedge T}+\epsilon \geq\epsilon > 0,\ \ \forall 0\leq{t}\leq{T}$. Since also $p<0$, it follows that
\begin{equation}\label{ineqYYY}
\underset{n}{\sup}(Y_{\tau_{n}\wedge T}+\epsilon)^{p}<\epsilon^{p}<+\infty.
\end{equation}
\indent
Remark $\ref{remapp}$ states that $A(t,s)\leq{0}, \ \forall 0\leq{t}\leq{s}\leq{T}$. Also $m(s)$, $\delta(s)$ are continuous functions and hence bounded on $[0,T]$, moreover, when $p<0$, we have $1-a(t,s)\hat{\Omega}(t)>0$ and $1-f(t,s)\hat{\Omega}(t)>0$ as well as $a(t,s)$, $b(t,s)$, $c(t,s)$, $f(t,s)$ and $g(t,s)$ are all bounded for $0\leq t\leq s\leq T$. We deduce that the explicit solutions $B(t,s)$ and $C(t,s)$ are both bounded on $0\leq t\leq s\leq T$ and hence
\begin{equation}
N(0,\eta)\leq{k_{1}\exp(k\eta)} \ \  \textrm{for some large constants}  \ \ k,k_{1}>1.\nonumber
\end{equation}
It follows that there exist some constants $\bar{k},\bar{k}_{1}>1$ such that
\begin{equation}
\underset{n}{\sup}\ N_{\tau_{n}\wedge T}^{1-p}\leq\underset{t\in[0,T]}{\sup}\Big(k_{1}\exp\Big(k\hat{\mu}_{t}\Big)\Big)^{1-p}\leq\bar{k}_{1}\exp\Big(\bar{k}\underset{t\in[0,T]}{\sup}\hat{\mu}_{t}\Big).\nonumber
\end{equation}
\indent
The Ornstein Uhlenbeck diffusion $\hat{\mu}_t$ satisfies $(\ref{newmu})$ which is equivalent to
\begin{equation}
\hat{\mu}_{t}=e^{-t\lambda}\eta+\bar{\mu}(1-e^{-t\lambda})+\int_{0}^{t}e^{\lambda(u-t)}\frac{\Big(\hat{\Omega}(u)+\sigma_{S}\sigma_{\mu}\rho\Big)}{\sigma_{S}}d\hat{W}_{u}.\nonumber
\end{equation}
Hence, there exists positive constants $l$ and $l_{1}>1$ large enough, such that:
\begin{equation}
\underset{t\in[0,T]}{\sup}\hat{\mu}_{t}\leq{l+\underset{t\in[0,T]}{\sup}l_{1}\hat{W}_{t}},\ \ \ t\in[0,T].\nonumber
\end{equation}
Using the distribution of running maximum of the Brownian Motion, there exists some positive constants $\bar{l}>1$ and $\bar{l}_{1}$ such that
\begin{equation}\label{ineqNNN}
\mathbb{E}\Big[\underset{n}{\sup}N_{\tau_{n}\wedge T}^{1-p}\Big]\leq{\bar{l}_{1}\mathbb{E}\Big[\exp\big(\underset{t\in[0,T]}{\sup}\bar{l}\hat{B}_{t}\Big)\Big]}<+\infty.
\end{equation}
\indent
At last, by $(\ref{ineqYYY})$ and $(\ref{ineqNNN})$, we can conclude that
\begin{equation}
\mathbb{E}\Big[\ \underset{n}{\sup}\ V(\tau_{n}\wedge T,X_{\tau_{n}\wedge T}+\epsilon,Z_{\tau_{n}\wedge T},\hat{\mu}_{\tau_{n}\wedge T})\Big]<+\infty.\nonumber
\end{equation}
Dominated Convergence Theorem gives 
\begin{equation}
\begin{split}
\underset{n\rightarrow+\infty}{\lim}\mathbb{E}\Big[V(\tau_{n}\wedge T,X_{\tau_{n}\wedge T}+\epsilon,Z_{\tau_{n}\wedge T},\hat{\mu}_{\tau_{n}\wedge T})\Big] & =
\mathbb{E}\Big[\frac{1}{p}(Y_{T}+\epsilon)^{p}N(T,\hat{\mu}_{T})\Big]\\
& =\mathbb{E}\Big[\frac{(X_{T}+\epsilon)^{p}}{p}\Big]>\mathbb{E}\Big[\frac{X_{T}^{p}}{p}\Big].\nonumber
\end{split}
\end{equation}
Combining this with equation $(\ref{equnV})$ and since $(\pi_{t}, c_{t})\in\mathcal{A}$, it follows that
\begin{equation}
V(0,x_{0}+\epsilon,z_{0},\eta_{0};\theta_{0})\geq{\underset{\pi,c\in\mathcal{A}}{\sup}\mathbb{E}\Big[\int_{0}^{T}\frac{(c_{s}-Z_{s})^{p}}{p}ds+\frac{X_{T}^{p}}{p}\Big]}=v(x_{0},z_{0},\eta_{0},\theta_{0}).\nonumber
\end{equation}
Notice $V(t,x,z,\eta;\theta_{0})$ is continuous in variable $x$, and since $\epsilon>0$ is arbitrary, we can take the limit and deduce that
\begin{equation}
V(0,x_{0},z_{0},\eta_{0};\theta_{0})=\underset{\epsilon\rightarrow 0}{\lim}V(0,x_{0}+\epsilon,z_{0},\eta_{0})\geq{v(x_{0},z_{0},\eta_{0},\theta_{0})}.\nonumber
\end{equation}
\indent
On the other hand, for $\pi^{\ast}_{t}$ and $c^{\ast}_{t}$ defined by $(\ref{optpi})$ and $(\ref{optc})$ respectively, we first need to show that the SDE for wealth process:
\begin{equation}\label{sdeX}
dX_{t}^{\ast}=(\pi^{\ast}_{t}\mu_{t}-c^{\ast}_{t})dt+\sigma_{S}\pi^{\ast}_{t}d\hat{W}_{t},\ \ \ 0\leq{t}\leq{T},
\end{equation}
with initial condition $x_{0}>m(0)z_{0}$ admits a unique strong solution which satisfies the constraint $X_{t}^{\ast}>m(t)Z_{t}^{\ast},\ \forall 0\leq t\leq T$.\\
\indent
Denote $Y_{t}^{\ast}=X_{t}^{\ast}-m(t)Z^{\ast}_{t}$, It\^{o}'s lemma implies that
\begin{equation}\label{eqnY}
\begin{split}
dY_{t}^{\ast}= & \Big[\pi^{\ast}_{t}\hat{\mu}_{t}-c^{\ast}_{t}-m_{t}(t)Z^{\ast}_{t}-m(t)\delta(t)c^{\ast}_{t}+m(t)\alpha(t)Z^{\ast}_{t}\Big]dt+\pi^{\ast}_{t}\sigma_{S}d\hat{W}_{t}\\
= & \Big[\Big(-m_{t}(t)+m(t)\alpha(t)\Big)Z^{\ast}_{t}-(1+m(t)\delta(t))c^{\ast}_{t}+\frac{\hat{\mu}_{t}^{2}}{(1-p)\sigma_{S}^{2}}Y_{t}^{\ast}\\
+ & \frac{\Big(\hat{\Omega}(t)+\sigma_{S}\sigma_{\mu}\rho\Big)}{\sigma^{2}_{S}}\frac{N_{\eta}}{N}\hat{\mu}_{t}Y_{t}^{\ast}\Big]dt+ \Big[\frac{\hat{\mu}_{t}}{(1-p)\sigma_{S}}+\frac{\Big(\hat{\Omega}(t)+\sigma_{S}\sigma_{\mu}\rho\Big)}{\sigma_{S}}\frac{N_{\eta}}{N}\Big]Y_{t}^{\ast}d\hat{W}_{t}.
\end{split}
\end{equation}
\indent
Using the definition of $m(t)$ by $(\ref{www})$ and substituting $c^{\ast}_{t}$ defined by $(\ref{optc})$ into $(\ref{eqnY})$ above, we will further have
\begin{equation}
\begin{split}
dY_{t}^{\ast}= & \Big[-\frac{\Big(1+\delta(t)m(t)\Big)^{\frac{-p}{1-p}}}{N}+\frac{\hat{\mu}_{t}^{2}}{(1-p)\sigma^{2}_{S}}+\frac{\Big(\hat{\Omega}(t)+\sigma_{S}\sigma_{\mu}\rho\Big)}{\sigma^{2}_{S}}\frac{N_{\eta}}{N}\hat{\mu}_{t}\Big]Y_{t}^{\ast}dt\\
& +\Big[\frac{\hat{\mu}_{t}}{(1-p)\sigma_{S}}+\frac{\Big(\hat{\Omega}(t)+\sigma_{S}\sigma_{\mu}\rho\Big)}{\sigma_{S}}\frac{N_{\eta}}{N}\Big]Y_{t}^{\ast}d\hat{W}_{t}.\nonumber
\end{split}
\end{equation}
\indent
In order to solve $X_{t}^{\ast}$ in a more explicit formula, we define the auxiliary process by
\begin{equation}
\Gamma_{t}=\frac{N(t,\hat{\mu}_{t})}{Y_{t}^{\ast}},\ \ \ \forall 0\leq{t}\leq{T}.\nonumber
\end{equation}
It\^{o}'s lemma implies that
\begin{equation}\label{eqngamma}
\begin{split}
d\Gamma_{t}=& \frac{\Gamma_{t}}{N}\Big[N_{t}-\lambda(\hat{\mu}_{t}-\bar{\mu})N_{\eta}+\frac{\Big(\hat{\Omega}(t)+\sigma_{S}\sigma_{\mu}\rho\Big)^{2}}{2\sigma^{2}_{S}}N_{\eta\eta}+\frac{\hat{\mu}_{t}\Big(\hat{\Omega}(t)+\sigma_{S}\sigma_{\mu}\rho\Big)p}{(1-p)\sigma^{2}_{S}}N_{\eta}\\
& +\Big(1+\delta(t)m(t)\Big)^{\frac{-p}{1-p}}+\frac{p\hat{\mu}_{t}^{2}}{(1-p)^{2}\sigma^{2}_{S}}N\Big]dt+\Gamma_{t}\Big[\frac{-\hat{\mu}_{t}}{(1-p)\sigma_{S}}\Big]d\hat{W}_{t}.
\end{split}
\end{equation}
\indent
Since $N(t,\eta)$ satisfies the linear PDE $(\ref{eqnN})$, $(\ref{eqngamma})$ is simplified as
\begin{equation}
d\Gamma_{t}=\Gamma_{t}\Big[\frac{p\hat{\mu}_{t}^{2}}{2(1-p)^{2}\sigma^{2}_{S}}\Big]dt+\Gamma_{t}\Big[\frac{-\hat{\mu}_{t}}{(1-p)\sigma_{S}}\Big]d\hat{W}_{t}.\nonumber
\end{equation}
Hence, the existence of the unique strong solution of the above SDE is guaranteed 
\begin{equation}
\Gamma_{t}=\Gamma_{0}\exp\Big(-\int_{0}^{t}\frac{\hat{\mu}_{u}^{2}}{2(1-p)\sigma^{2}_{S}}du-\int_{0}^{t}\frac{\hat{\mu}_{u}}{(1-p)\sigma_{S}}d\hat{W}_{u}\Big),\nonumber
\end{equation}
\indent
Initial condition $\Gamma_{0}=\frac{N(0,\eta)}{x_{0}-m(0)z_{0}}>0$ implies that $\Gamma_{t}>0,\ \ \forall \ 0\leq{t}\leq{T}$. Therefore, we finally proved that the SDE $(\ref{sdeX})$ has a unique strong solution defined by $(\ref{optiwealth})$ and the solution $X_{t}^{\ast}$ satisfies the wealth process constraint $(\ref{kyu})$\\
\indent
Next, we proceed to verify the pair $(\pi^{\ast}_{t},c^{\ast}_{t})$ is indeed in the admissible space \ $\mathcal{A}$. \\
\indent
First, by the definition $(\ref{optpi})$ and $(\ref{optc})$, it's clear that $\pi^{\ast}_{t}$ and $c^{\ast}_{t}$ are $\mathcal{F}_{t}^{S}$ progressively measurable, and by the path continuity of $Y_{t}^{\ast}=X_{t}^{\ast}-m(t)Z_{t}^{\ast}$, hence, of $\pi^{\ast}_{t}$ and $c^{\ast}_{t}$, it's easy to show that
\begin{equation}
\int_{0}^{T}(\pi^{\ast}_{t})^{2}dt<+\infty,\ \  \ \ and\ \ \ \ \int_{0}^{T}c^{\ast}_{t}dt<+\infty,\ \ a.s.\nonumber
\end{equation}
\indent
Also, since $X_{t}^{\ast}>m(t)Z_{t}^{\ast},\ \forall t\in[0,T]$, by the definition of $c^{\ast}_{t}$, the consumption constraint $c^{\ast}_{t}>Z^{\ast}_{t}, \ \ \forall t\in[0,T]$ is satisfied. It follows that $(\pi^{\ast}_{t}, c^{\ast}_{t})\in\mathcal{A} $.\\
\indent
Given the pair of control policy $(\pi^{\ast}_{t},c^{\ast}_{t})$ as above, instead of $(\ref{inqV})$, the equality is verified
\begin{equation}
V(0,x_{0},z_{0},\eta_{0};\theta_{0})=\mathbb{E}\Big[\int_{0}^{\tau_{n}\wedge T}\frac{(c^{\ast}_{t}-Z^{\ast}_{t})^{p}}{p}dt\Big]+
\mathbb{E}\Big[V(\tau_{n}\wedge T,X^{\ast}_{\tau_{n}\wedge T},Z^{\ast}_{\tau_{n}\wedge T},\hat{\mu}_{\tau_{n}\wedge T})\Big].\nonumber
\end{equation}
Monotone Convergence Theorem implies
\begin{equation}
\lim_{n\rightarrow+\infty}\mathbb{E}\Big[\int_{0}^{\tau_{n}\wedge T}\frac{(c^{\ast}_{t}-Z^{\ast}_{t})^{p}}{p}dt\Big]=\mathbb{E}\Big[\int_{0}^{T}\frac{(c^{\ast}_{t}-Z^{\ast}_{t})^{p}}{p}dt \Big],\nonumber
\end{equation}
Moreover, when $p<0$, we have that the function $V(t,x,z,\eta)<0$ by it's definition. Fatou's lemma deduces that
\begin{equation}
\underset{n\rightarrow+\infty}{\lim\sup}\mathbb{E}\Big[V(\tau_{n}\wedge T,X^{\ast}_{\tau_{n}\wedge T},Z^{\ast}_{\tau_{n}\wedge T},\hat{\mu}_{\tau_{n}\wedge T})\Big]\leq{}\mathbb{E}\Big[V(T,X^{\ast}_{T},Z^{\ast}_{T},\hat{\mu}_{T})\Big]=\mathbb{E}\Big[\frac{(X^{\ast}_{T})^{p}}{p}\Big].\nonumber
\end{equation}
\indent
Therefore, it follows that
\begin{equation}
V(0,x_{0},z_{0},\eta_{0};\theta_{0})\leq{\mathbb{E}\Big[\int_{0}^{T}\frac{(c^{\ast}_{t}-Z^{\ast}_{t})^{p}}{p}dt+\frac{(X^{\ast}_{T})^{p}}{p}\Big]\leq v(x_{0},z_{0},\eta_{0},\theta_{0})}\nonumber
\end{equation}
which completes the proof. 
\end{proof}

\subsection{The Case: $\mathit{0<p<1}$}
The following two Lemmas play important roles in the proof of the second part of our main result.
\begin{lemma}\label{prop62}
If constant $k>0$ satisfies:
\begin{equation}\label{cond2}
k<\frac{\lambda^{2}\sigma^{2}_{S}}{2(\Theta+\sigma_{S}\sigma_{\mu}\rho)^{2}}
\end{equation}
for any $t\geq 0$, there exists a constant $\Lambda_{1}$ such that
\begin{equation}
\mathbb{E}\Big[\exp\Big(\int_{0}^{t}k\hat{\mu}^{2}_{s}ds\Big)\Big]\leq \Lambda_{1}<+\infty.\nonumber
\end{equation}
\end{lemma}
\begin{proof}. Similar to the proof of Lemma $12$ of \cite{MR2086171}. It is easy to choose an increasing sequence of smooth functions $Q_{n}(y)\nearrow ky^{2}$ as $n\rightarrow\infty$ such that $0\leq{Q_{n}(y)}\leq{n}$ with $\Big|Q_{n}'(y)\Big|$ and $\Big|Q_{n}''(y)\Big|$ uniformly bounded. And for each fixed $t\geq 0$ and $\eta$, we define:
\begin{equation}
\phi(t,\eta)=\mathbb{E}\Big[\exp\Big(\int_{0}^{t}Q_{n}(\hat{\mu}_{s})ds\Big)\Big],\nonumber
\end{equation}
where $\hat{\mu}_{0}=\eta$.\\
\indent
Similar to the proof of Feynman-Kac formula, the function $\phi(t,\eta)$ is a classical solution of the linear parabolic equation:
\begin{equation}\label{fphi}
\phi_{t}=\frac{\Big(\hat{\Omega}(t)+\sigma_{S}\sigma_{\mu}\rho\Big)^{2}}{2\sigma^{2}_{S}}\phi_{\eta\eta}-\lambda(\eta-\bar{\mu})\phi_{\eta}+Q_{n}(\eta)\phi,
\end{equation}
with initial condition $\phi(0,\eta)=1$. See also Lemma $1.12$ in \cite{pangt} for details.\\
\indent
First, it's clear that constant $0$ is a subsolution of the above equation. Moreover, under assumption $(\ref{cond2})$, it's easy to show that for each fixed $t\geq 0$, the equation:
\begin{equation}
\frac{2\Big(\hat{\Omega}(t)+\sigma_{S}\sigma_{\mu}\rho\Big)^{2}}{\sigma^{2}_{S}}x^{2}-2\lambda x+k=0\nonumber
\end{equation}
has two positive real roots
\begin{equation}
x_{1}=\frac{\lambda-\sqrt{\lambda^{2}-\frac{2(\hat{\Omega}(t)+\sigma_{S}\sigma_{\mu}\rho)^{2}}{\sigma^{2}_{S}}k}}{2\frac{(\hat{\Omega}(t)+\sigma_{S}\sigma_{\mu}\rho)^{2}}{\sigma^{2}_{S}}},\ \ \textrm{and}\ \ x_{2}=\frac{\lambda+\sqrt{\lambda^{2}-\frac{2(\hat{\Omega}(t)+\sigma_{S}\sigma_{\mu}\rho)^{2}}{\sigma^{2}_{S}}k}}{2\frac{(\hat{\Omega}(t)+\sigma_{S}\sigma_{\mu}\rho)^{2}}{\sigma^{2}_{S}}}.\nonumber
\end{equation}
For any positive constant $a$ such that:
\begin{equation}
0<a<\frac{\lambda+\sqrt{\lambda^{2}-\frac{(\Theta_{2}+\sigma_{S}\sigma_{\mu}\rho)^{2}}{\sigma^{2}_{S}}k}}{2\frac{(\Theta_{1}+\sigma_{S}\sigma_{\mu}\rho)^{2}}{\sigma^{2}_{S}}},\nonumber
\end{equation}
with $\Theta_{1}=\max(\theta,\theta^{\ast})$ and $\Theta_{2}=\min(\theta,\theta^{\ast})$, and the positive constant $b$ such that:
\begin{equation}
b>a\frac{(\Theta_{1}+\sigma_{S}\sigma_{\mu}\rho)^{2}}{\sigma^{2}_{S}}-\frac{\lambda^{2}\bar{\mu}^{2}a^{2}}{2a^{2}\frac{(\Theta_{1}+\sigma_{S}\sigma_{\mu}\rho)^{2}}{\sigma^{2}_{S}}-2a\lambda+k}.\nonumber
\end{equation}
\indent
It's easy to verify that $f(t,\eta)=\exp(bt+a\eta^{2})$ satisfies
\begin{equation}
f_{t}\geq{}\frac{\Big(\hat{\Omega}(t)+\sigma_{S}\sigma_{\mu}\rho\Big)^{2}}{2\sigma^{2}_{S}}f_{\eta\eta}-\lambda(\eta-\bar{\mu})f_{\eta}+k\eta^{2},\nonumber
\end{equation}
with the initial condition $f(0,\eta)\geq{1}$.\\
\indent
Since $Q_{n}(\eta)<k\eta^{2}$, we get function $f(t,\eta)$ is the supersolution of the equation $(\ref{fphi})$, and $\big{\langle}0,f(t,\eta)\big{\rangle}$ is the coupled subsolution and supersolution. Theorem $7.2$ from \cite{Pao} shows that function $\phi(t,\eta)$ satisfies: $0\leq\phi(t,\eta)\leq f(t,\eta)\equiv{\Lambda_{1}}$, and hence Monotone Convergence Theorem leads to
\begin{equation}
\mathbb{E}\Big[\exp\Big(\int_{0}^{t}k\hat{\mu}^{2}_{s}ds\Big)\Big]\leq\Lambda_{1}<+\infty.\nonumber
\end{equation} 
\end{proof}

\begin{lemma}\label{prop63}
If constant $\bar{k}>0$ satisfies
\begin{equation}\label{parassm}
\bar{k}<\frac{\lambda\sigma_{S}^{2}}{(\Theta+\sigma_{S}\sigma_{\mu}\rho)^{2}},
\end{equation}
for any fixed constant $\kappa>0$, there exists a constant $\Lambda_{2}$ independent of $t$ such that
\begin{equation}
\mathbb{E}\Big[\exp\Big(\bar{k}(\hat{\mu}_{t}+\kappa)^{2}\Big)\Big]\leq \Lambda_{2}<\infty,\ \ \ t\in[0,T].\nonumber
\end{equation}
\end{lemma}
\begin{proof}.
Similar to the proof of Lemma $\ref{prop62}$, we again construct an increasing sequence of functions $\{Q_{n}(y)\}$ for $n\in\mathbb{N}$ such that
$\lim_{n\rightarrow+\infty}Q_{n}(y)=\bar{k}(y+\kappa)^{2}$. And for each fixed $t\in[0,T]$ and $\eta$, we define:
\begin{equation}
\psi(t,\eta)=\mathbb{E}\Big[\exp\Big(Q_{n}(\hat{\mu}_{t})\Big)\Big],\nonumber
\end{equation}
where $\hat{\mu}_{0}=\eta$.\\
\indent
As a direct corollary of Theorem $5.6.1$ of \cite{FR}, the function $\psi(t,\eta)$ is a classical solution of the linear parabolic equation
\begin{equation}\label{eeekkk}
\psi_{t}=\frac{\Big(\hat{\Omega}(t)+\sigma_{S}\sigma_{\mu}\rho\Big)^{2}}{2\sigma^{2}_{S}}\psi_{\eta\eta}-\lambda(\eta-\bar{\mu})\psi_{\eta},
\end{equation}
with initial condition $\psi(0,\eta)=e^{Q_{n}(\eta)}$.\\
\indent
Under the condition $(\ref{parassm})$, we can choose any constant $a$ such that
\begin{equation}
\bar{k}<a<\frac{\lambda\sigma_{S}^{2}}{(\Theta+\sigma_{S}\sigma_{\mu}\rho)^{2}},\nonumber
\end{equation}
and any positive constant b such that
\begin{equation}
\begin{split}
b> & a\frac{(\Theta+\sigma_{S}\sigma_{\mu}\rho)^{2}}{\sigma^{2}_{S}}+\frac{2a^{2}(\Theta+\sigma_{S}\sigma_{\mu}\rho)^{2}\kappa^{2}}{\sigma_{S}^{2}}+2a\lambda\bar{\mu}\kappa\\
& -\frac{\Big(\frac{2a^{2}\kappa(\Theta+\sigma_{S}\sigma_{\mu}\rho)^{2}}{\sigma_{S}^{2}}-a\lambda\kappa-a\lambda\bar{\mu}\Big)^{2}}{2a^{2}\frac{(\Theta+\sigma_{S}\sigma_{\mu}\rho)^{2}}{\sigma^{2}_{S}}-2a\lambda},\nonumber
\end{split}
\end{equation}
\indent
It is easy to verify that $f(t,\eta)=\exp(bt+a(\eta+\kappa)^{2})$ satisfies
\begin{equation}
f_{t}\geq{}\frac{\Big(\hat{\Omega}(t)+\sigma_{S}\sigma_{\mu}\rho\Big)^{2}}{2\sigma^{2}_{S}}f_{\eta\eta}-\lambda(\eta-\bar{\mu})f_{\eta},\nonumber
\end{equation}
with the initial condition $f(0,\eta)=e^{a(\eta+\kappa)^{2}}\geq \psi(0,\eta)$, hence we get the function $f(t,\eta)$ is the supersolution of the equation $(\ref{eeekkk})$, and it is trivial to show $g(t,\eta)\equiv 0$ is the subsolution, therefore $\big{\langle}0,f(t,\eta)\big{\rangle}$ are the coupled subsolution and supersolution. Again by Theorem $7.2$ from \textit{Pao}\ \cite{Pao}, that function $\psi(t,\eta)$ satisfies: $0\leq\psi(t,\eta)\leq f(t,\eta)\leq e^{bT+a(\eta+\kappa)^{2}}\equiv{\Lambda_{2}}$, hence Monotone Convergence Theorem implies that 
\begin{equation}
\mathbb{E}\Big[\exp\Big(\bar{k}(\hat{\mu}_{t}+\kappa)^{2}\Big)\Big]\leq\Lambda_{2}<+\infty,\ \ \forall t\in[0,T].\nonumber
\end{equation}
\end{proof}

\begin{proof}[(THE PROOF OF THEOREM $\ref{mainmain}$, CONTINUED)] \ \\
\indent
For any pair of admissible control $(\pi_{t},c_{t})\in\mathcal{A}$ and localizing sequence $\tau_n$, similar to the case for $p<0$, we have
\begin{equation}\label{ineq2}
V(0,x_{0},z_{0},\eta_{0})\geq{\mathbb{E}\Big[\int_{0}^{\tau_{n}\wedge T}\frac{(c_{s}-Z_{s})^{p}}{p}ds\Big]+\mathbb{E}\Big[V(\tau_{n}\wedge T,X_{\tau_{n}\wedge T},Z_{\tau_{n}\wedge T},\hat{\mu}_{\tau_{n}\wedge T})\Big]}.
\end{equation}
\indent
Monotone Convergence Theorem first implies that
\begin{equation}
\lim_{n\rightarrow +\infty}\mathbb{E}\Big[\int_{0}^{\tau_{n}\wedge T}\frac{(c_{s}-Z_{s})^{p}}{p}ds\Big]=\mathbb{E}\Big[\int_{0}^{T}\frac{(c_{s}-Z_{s})^{p}}{p}ds\Big].\nonumber
\end{equation}
\indent
For the case $0<p<1$, $V(t,x,z,\eta)\geq 0$ for all $t\in[0,T]$ by the definition $(\ref{solut})$ and $(\ref{kyu})$. Fatou's lemma yields that
\begin{equation}
\lim_{n\rightarrow+\infty}\mathbb{E}\Big[V(\tau_{n}\wedge T,X_{\tau_{n}\wedge T},Z_{\tau_{n}\wedge T},\hat{\mu}_{\tau_{n}\wedge T})\Big]\geq{\mathbb{E}\Big[V(T,X_{T},Z_{T},\hat{\mu}_{T})\Big]}=\mathbb{E}\Big[\frac{X_{T}^{p}}{p}\Big],\nonumber
\end{equation}
which implies that:
\begin{equation}
V(0,x_{0},z_{0},\eta_{0})\geq{\underset{\pi,c\in\mathcal{A}}{\sup}\mathbb{E}\Big[\int_{0}^{T}\frac{(c_{s}-Z_{s})^{p}}{p}ds+\frac{X_{T}^{p}}{p}\Big]}=v(x_{0},z_{0},\eta_{0},\theta_{0}).\nonumber
\end{equation}
\indent
On the other hand, for $\pi^{\ast}_{t}$ and $c^{\ast}_{t}$ defined by $(\ref{optpi})$, $(\ref{optc})$, similar to the case $p<0$, we can show the pair $(\pi^{\ast}_{t},c^{\ast}_{t})$ is in the admissible space \ $\mathcal{A}$. \\
\indent
For the given policies $\pi^{\ast}_{t}$ and $c^{\ast}_{t}$, similarly, we can get the equality:
\begin{equation}
V(0,x_{0},z_{0},\eta_{0})=\mathbb{E}\Big[\int_{0}^{\tau_{n}\wedge T}\frac{(c_{s}^{\ast}-Z_{s}^{\ast})^{p}}{p}ds\Big]+\mathbb{E}\Big[V(\tau_{n}\wedge T,X_{\tau_{n}\wedge T}^{\ast},Z_{\tau_{n}\wedge T}^{\ast},\hat{\mu}_{\tau_{n}\wedge T})\Big].\nonumber
\end{equation}
The definition of $V(t,x,z,\eta)$ deduces that
\begin{equation}\label{now}
V(T\wedge\tau_{n},X_{T\wedge\tau_{n}}^{\ast},Z_{T\wedge\tau_{n}}^{\ast},\hat{\mu}_{T\wedge \tau_{n}})\leq k_{1}\Big[\Big(\frac{Y^{\ast}}{N}\Big)^{2p}_{T\wedge\tau_{n}}+N^{2}(T\wedge\tau_{n},\hat{\mu}_{T\wedge\tau_{n}})\Big]
\end{equation}
for some positive constants $k_{1}$, which are independent of $n$.\\
\indent
For the first term of (\ref{now}), we notice that
\begin{equation}
\begin{split}
\Big(\frac{Y^{\ast}}{N}\Big)^{2p}_{T\wedge\tau_{n}}\leq & \frac{1}{2}\Big[\exp\Big(\int_{0}^{T\wedge\tau_{n}}\frac{4p\hat{\mu}_{u}}{(1-p)\sigma_{S}}d\hat{W}_{u}-\int_{0}^{T\wedge\tau_{n}}\frac{4p^{2}\hat{\mu}^{2}_{u}}{(1-p)^{2}\sigma^{2}_{S}}du\Big)\\
& +\exp\Big(\int_{0}^{T\wedge\tau_{n}}\frac{2(2p^{2}+p)\hat{\mu}^{2}_{u}}{(1-p)^{2}\sigma^{2}_{S}}du\Big)\Big],\nonumber
\end{split}
\end{equation}
hence, we have
\begin{equation}
\begin{split}
\mathbb{E}\Big[\sup_{n}\Big(\frac{Y^{\ast}}{N}\Big)^{2p}_{T\wedge\tau_{n}}\Big]\leq & \frac{1}{2}\mathbb{E}\Big[\sup_{n}\exp\Big(\int_{0}^{T\wedge\tau_{n}}\frac{4p\hat{\mu}_{u}}{(1-p)\sigma_{S}}d\hat{W}_{u}-\int_{0}^{T\wedge\tau_{n}}\frac{4p^{2}\hat{\mu}^{2}_{u}}{(1-p)^{2}\sigma^{2}_{S}}du\Big)\\
& +\sup_{n}\exp\Big(\int_{0}^{T\wedge\tau_{n}}\frac{2(2p^{2}+p)\hat{\mu}^{2}_{u}}{(1-p)^{2}\sigma^{2}_{S}}du\Big)\Big].\nonumber
\end{split}
\end{equation}
Again since $\hat{\mu}_{t}$ follows the dynamics $(\ref{newmu})$, by Bene\v{s}' condition (see Corollary $3.5.14$ and Corollary $3.5.16$ in \cite{MR917065}), we see that the exponential local martingale $M_{t}=\exp\Big(\int_{0}^{t}\frac{2p\hat{\mu}_{u}}{(1-p)\sigma_{S}}d\hat{W}_{u}-\int_{0}^{t}\frac{2p^{2}\hat{\mu}^{2}_{u}}{(1-p)^{2}\sigma^{2}_{S}}du\Big)$ is a true martingale, and hence, by Doob's maximal inequality, we first derive that
\begin{equation}
\begin{split}
& \mathbb{E}\Big[\sup_{n}\exp\Big(\int_{0}^{T\wedge\tau_{n}}\frac{4p\hat{\mu}_{u}}{(1-p)\sigma_{S}}d\hat{W}_{u}-\int_{0}^{T\wedge\tau_{n}}\frac{4p^{2}\hat{\mu}^{2}_{u}}{(1-p)^{2}\sigma^{2}_{S}}du\Big)\Big]\\
\leq & \mathbb{E}\Big[\sup_{t\in[0,T]}\exp\Big(\int_{0}^{t}\frac{4p\hat{\mu}_{u}}{(1-p)\sigma_{S}}d\hat{W}_{u}-\int_{0}^{t}\frac{4p^{2}\hat{\mu}^{2}_{u}}{(1-p)^{2}\sigma^{2}_{S}}du\Big)\Big]\\
\leq & k(p)\mathbb{E}\Big[\exp\Big(\int_{0}^{T}\frac{4p\hat{\mu}_{u}}{(1-p)\sigma_{S}}d\hat{W}_{u}-\int_{0}^{T}\frac{4p^{2}\hat{\mu}^{2}_{u}}{(1-p)^{2}\sigma^{2}_{S}}du\Big)\Big]\nonumber
\end{split}
\end{equation}
where $k(p)$ is a constant depending on $p$. Moreover, similar to the proofs of Corollary $3.5.14$ and Corollary $3.5.16$ in \cite{MR917065}, Corollary $1$ and Corollary $2$ in \cite{MR2003243} further states that the true martingale $M_{t}$ defined as above satisfies the finite moments property, i.e., for any $r>1$, we have $\mathbb{E}\Big[M^{r}_{T}\Big]<\infty$. Therefore it follows that for $r=2$,
\begin{equation}
\mathbb{E}\Big[\exp\Big(\int_{0}^{T}\frac{4p\hat{\mu}_{u}}{(1-p)\sigma_{S}}d\hat{W}_{u}-\int_{0}^{T}\frac{4p^{2}\hat{\mu}^{2}_{u}}{(1-p)^{2}\sigma^{2}_{S}}du\Big)\Big]<\infty.\nonumber
\end{equation}
\indent
For the second part of (\ref{now}), we can apply Assumption $\ref{ass222}$ and Lemma $\ref{prop62}$ to deduce that:
\begin{equation}
\begin{split}
\mathbb{E}\Big[\sup_{n} \exp\Big(\int_{0}^{T\wedge\tau_{n}}\frac{(2p^{2}+2p)\hat{\mu}^{2}_{u}}{(1-p)^{2}\sigma^{2}_{S}}du\Big) \Big] & \leq \mathbb{E}\Big[ \exp\Big(\int_{0}^{T}\frac{(2p^{2}+2p)\hat{\mu}^{2}_{u}}{(1-p)^{2}\sigma^{2}_{S}}du\Big) \Big]\\
& <\Lambda_{1}<+\infty,\nonumber
\end{split}
\end{equation}
for some constant $\Lambda_{1}>0$.\\
\indent
Under the condition that $a(t,s)$, $b(t,s)$, $c(t,s)$, $f(t,s)$ and $g(t,s)$ defined in Lemma $\ref{brene}$ are bounded and $1-a(t,s)\hat{\Omega}(t)\neq 0$ on $0\leq t\leq s\leq T$, functions $A(t,s)$, $B(t,s)$ and $C(t,s)$ are all bounded on $0\leq t\leq s\leq T$. As a simple consequence, there exist constants $k$, $k_1$ such that
\begin{equation}
N(t,\eta)\leq ke^{\bar{K}_{1}(\eta+k_{1})^{2}},\nonumber
\end{equation}
where $A(t,s)\leq \bar{K}_{1}$ for all $0\leq t\leq s\leq T$, and hence we have
\begin{equation}
\underset{n}{\sup}\Big(N^{2}(T\wedge\tau_{n},\hat{\mu}_{T\wedge\tau_{n}})\Big)\leq \underset{t\in[0,T]}{\sup}k^2e^{2\bar{K}_{1}(\hat{\mu}_{t}+k_{1})^{2}}.\nonumber
\end{equation}
\indent
So it is enough to show that
\begin{equation}\label{bbbbd}
\mathbb{E}\Big[\underset{t\in[0,T]}{\sup}e^{2\bar{K}_{1}(\hat{\mu}_{t}+k_{1})^{2}}\Big]<\infty.
\end{equation}
Define $\varphi(x)\triangleq e^{2\bar{K}_{1}(x+k_{1})^{2}}$, It\^{o}'s lemma gives
\begin{equation}
\begin{split}
d\varphi(\hat{\mu}_{t})= & \varphi(\hat{\mu}_{t})\Big[\Big(-4\bar{K}_{1}\lambda+8\bar{K}_{1}^{2}\frac{\Big(\hat{\Omega}(t)+\sigma_{S}\sigma_{\mu}\rho\Big)^{2}}{\sigma_{S}^{2}}\Big)\hat{\mu}_{t}^{2}+\Big(4\bar{K}_1\lambda\bar{\mu}-4\lambda\bar{K}_1k_1\\
&+16\bar{K}_1^2k_1\frac{\Big(\hat{\Omega}(t)+\sigma_S\sigma_{\mu}\rho\Big)^2}{\sigma_S^2}\Big)\hat{\mu}_t+2\bar{K}_1\frac{\Big(\hat{\Omega}(t)+\sigma_S\sigma_{mu}\rho\Big)^2}{\sigma_S^2}+4\lambda\bar{K}_1k_1\bar{\mu}\\
&+8\bar{K}_1^2k_1^2\frac{\Big(\hat{\Omega}(t)+\sigma_S\sigma_{\mu}\rho\Big)^2}{\sigma_S^2}\Big]dt+dL_t.\nonumber
\end{split}
\end{equation}
The condition $(\ref{newc})$ guarantees $-4\bar{K}_{1}\lambda+8\bar{K}_{1}^{2}\frac{\Big(\hat{\Omega}(t)+\sigma_{S}\sigma_{\mu}\rho\Big)^{2}}{\sigma_{S}^{2}}<0$, and hence there exists an upper bound constant $k_{2}>0$ such that
\begin{equation}
d\varphi(\hat{\mu}_{t})\leq \varphi(\hat{\mu}_{t})k_{2}dt+dL_{t},\nonumber
\end{equation}
where the local martingale part is
\begin{equation}
dL_{t}\triangleq \varphi(\hat{\mu}_{t})4\bar{K}_{1}\hat{\mu}_{t}\frac{\Big(\hat{\Omega}(t)+\sigma_{S}\sigma_{\mu}\rho\Big)}{\sigma_{S}}d\hat{W}_{t}.\nonumber
\end{equation}
It follows that
\begin{equation}
\mathbb{E}\Big[\underset{t\in[0,T]}{\sup}\varphi(\hat{\mu}_{t})\Big]\leq \varphi(\eta)+\int_{0}^{T}k_{2}\mathbb{E}\Big[\underset{s\in[0,t]}{\sup}\varphi(\hat{\mu}_{s})\Big]dt+\mathbb{E}\Big[\underset{t\in[0,T]}{\sup}L_{t}\Big],\nonumber
\end{equation}

Burhholder-Davis-Gundy Inequality and Jensen's Inequality induce that
\begin{equation}
\mathbb{E}\Big[\underset{t\in[0,T]}{\sup}L_{t}\Big]\leq k_{3}\Big(\int_{0}^{T}\mathbb{E}\Big[16\bar{K}_{1}^{2}\frac{\Big(\hat{\Omega}(t)+\sigma_{S}\sigma_{\mu}\rho\Big)^{2}}{\sigma_{S}^{2}}\hat{\mu}_{t}^{2}e^{4\bar{K}_{1}(\hat{\mu}_{t}+k_{1})^{2}}\Big]dt\Big)^{\frac{1}{2}},\nonumber
\end{equation}
for some constant $k_3$.

Under the condition $(\ref{newc})$, there exists a constant $\epsilon>0$ such that $4\bar{K}_{1}+\epsilon<\frac{\lambda\sigma_{S}^{2}}{(\Theta+\sigma_{S}\sigma_{\mu}\rho)^{2}}$. By H\"{o}lder's Inequality and choosing the conjugates $q=\frac{4\bar{K}_{1}+\epsilon}{4\bar{K}_{1}}$ and $\frac{1}{q}+\frac{1}{q'}=1$, it follows that 
\begin{equation}
\mathbb{E}\Big[\hat{\mu}_{t}^{2}e^{4\bar{K}_{1}(\hat{\mu}_{t}+k_{1})^{2}}\Big]\leq \Big(\mathbb{E}\Big[\hat{\mu}_{t}^{2q'}\Big]\Big)^{\frac{1}{q'}}\Big(\mathbb{E}\Big[e^{(4\bar{K}_{1}+\epsilon)(\hat{\mu}_{t}+k_{1})^{2}}\Big]\Big)^{\frac{1}{q}}.\nonumber
\end{equation}
Lemma $\ref{prop63}$ implies the existence of a constant $\Lambda_{2}$ independent of $t$ such that
\begin{equation}
\mathbb{E}\Big[e^{(4\bar{K}_{1}+\epsilon)(\hat{\mu}_{t}+k_{1})^{2}}\Big]\leq \Lambda_{2}<\infty, \ \forall t\in[0,T],\nonumber
\end{equation}
Again by the fact that there exist positive constants $l$ and $l_{1}>1$ large enough such that
\begin{equation}
\underset{t\in[0,T]}{\sup}\hat{\mu}_{t}\leq{l+\underset{t\in[0,T]}{\sup}l_{1}\hat{W}_{t}},\ \ \ t\in[0,T],\nonumber
\end{equation}
we can obtain
\begin{equation}
\int_{0}^{T}\Big(\mathbb{E}\Big[\hat{\mu}_{t}^{2q'}\Big]\Big)^{\frac{1}{q'}}dt\leq T\Big(\mathbb{E}\Big[(l+\underset{t\in[0,T]}{\sup}l_{1}\hat{W}_{t})^{2q'}\Big]\Big)^{\frac{1}{q'}}<\infty\nonumber
\end{equation}
due to the distribution of running maximum of the Brownian motion $\hat{W}_{t}$. Hence we get the boundedness of $\mathbb{E}\Big[\underset{t\in[0,T]}{\sup}L_{t}\Big]\leq k_{4}<\infty$ for some constant $k_{4}$, and
\begin{equation}
\mathbb{E}\Big[\underset{t\in[0,T]}{\sup}\varphi(\hat{\mu}_{t})\Big]\leq \varphi(\eta)+\int_{0}^{T}k_{2}\mathbb{E}\Big[\underset{s\in[0,t]}{\sup}\varphi(\hat{\mu}_{s})\Big]dt+k_{4}.\nonumber
\end{equation}
The Gronwall's Inequality verifies $(\ref{bbbbd})$.\\
\indent
Therefore, putting all pieces together, we eventually derived that
\begin{equation}
\mathbb{E}\Big[\sup_{n}V(\tau_{n}\wedge T,X_{\tau_{n}\wedge T},Z_{\tau_{n}\wedge T},\hat{\mu}_{\tau_{n}\wedge T})\Big]<\infty
\end{equation}
and Dominated Convergence Theorem leads to
\begin{equation}
\lim_{n\rightarrow\infty}\mathbb{E}\Big[V(\tau_{n}\wedge T,X_{\tau_{n}\wedge T},Z_{\tau_{n}\wedge T},\hat{\mu}_{\tau_{n}\wedge T})\Big]=\mathbb{E}\Big[\frac{(X^{\ast}_{T})^{p}}{p}\Big].\nonumber
\end{equation}
Together with Monotone Convergence Theorem, we deduce
\begin{equation}
V(0,x_{0},z_{0},\eta_{0};\theta_{0})=\mathbb{E}\Big[\int_{0}^{T}\frac{(c^{\ast}_{s}-Z^{\ast}_{s})^{p}}{p}ds+\frac{(X^{\ast}_{T})^{p}}{p}\Big]\leq{v(x_{0},z_{0},\eta_{0},\theta_{0})},\nonumber
\end{equation}
which completes the proof.
\end{proof}

\appendix
\section{Fully Explicit Solutions to The Auxiliary ODEs}\label{ApA}
Following the arguments by \cite{Kim}, we can even solve the auxiliary ODEs $(\ref{eqna})$, $(\ref{eqnb})$, $(\ref{eqnc})$, $(\ref{eqnL})$ and $(\ref{eqnF})$ fully explicitly depending on the risk aversion constant $p$ and all the market coefficients $\sigma_{S}$, $\sigma_{\mu}$, $\lambda$, $\rho$:
\subsection{The Normal Solution}
The condition for the Normal solution is
\begin{equation}\label{confffff}
\Delta\triangleq\lambda^{2}-\frac{2\lambda p\rho\sigma_{\mu}}{(1-p)\sigma^{2}_{S}}-\frac{p\sigma^{2}_{\mu}}{(1-p)\sigma^{2}_{S}}>0,
\end{equation}
and then we define
\begin{equation}
\xi \triangleq\sqrt{\Delta}=\sqrt{\gamma^{2}_{2}-\gamma_{1}\gamma_{3}},\ \ \ \ \ \gamma_{1}\triangleq\frac{(1-p+p\rho^{2})}{1-p}\sigma_{\mu}^{2},\nonumber
\end{equation}
\begin{equation}
\gamma_{2} \triangleq -\lambda+\frac{p\rho\sigma_{\mu}}{(1-p)\sigma_{S}},\ \ \ \ \ \ \ \ \gamma_{3}\triangleq \frac{p}{(1-p)\sigma^{2}_{S}},\nonumber
\end{equation}
\begin{equation}
\xi_{1} \triangleq \frac{\sqrt{(1-\rho^{2})\sigma^{2}_{\mu}+(\lambda\sigma_{S}+\rho\sigma_{\mu})^{2}}}{\sigma_{S}}.\nonumber
\end{equation}
\indent
We can solve the equations  $(\ref{eqna})$, $(\ref{eqnb})$, $(\ref{eqnc})$, $(\ref{eqnL})$ and $(\ref{eqnF})$ as:
\begin{equation}
\begin{split}
a(t,s)= & \frac{p(1-e^{2\xi(t-s)})}{2(1-p)\sigma_{S}^{2}\Big[2\xi-(\xi+\gamma_{2})(1-e^{2\xi(t-s)})\Big]},\\
b(t,s)= & \frac{p\lambda\bar{\mu}(1-e^{\xi(t-s)})^{2}}{2(1-p)\sigma_{S}^{2}\xi\Big[2\xi-(\xi+\gamma_{2})(1-e^{2\xi(t-s)})\Big]},\\
c(t,s)= & \frac{p}{2(1-p)\sigma_{S}^{2}}\Big(\frac{\lambda^{2}\bar{\mu}^{2}}{\xi^{2}}-\frac{\sigma_{\mu}^{2}}{\xi+\gamma_{2}}\Big)(s-t)\\
& +\frac{p\lambda^{2}\bar{\mu}^{2}\Big[(\xi+2\gamma_{2})e^{2\xi(t-s)}-4\gamma_{2}e^{\xi(t-s)}+2\gamma_{2}-\xi\Big]}{2(1-p)\sigma_{S}^{2}\xi^{3}\Big[2\xi-(\xi+\gamma_{2})\Big(1-e^{2\xi(t-s)}\Big)\Big]}\\
& -\frac{p\sigma_{\mu}^{2}}{2(1-p)\sigma_{S}^{2}(\xi^{2}-\gamma_{2}^{2})}\log\Big|\frac{2\xi-(\xi+\gamma_{2})(1-e^{2\xi(t-s)})}{2\xi}\Big|,\nonumber
\end{split}
\end{equation}
\begin{equation}
\begin{split}
f(t,s)= & -\frac{1}{2\sigma_{S}}\frac{1-e^{2\xi_{1}(t-s)}}{(\sigma_{S}\xi_{1}+\lambda\sigma_{S}+\rho\sigma_{\mu})+(\sigma_{S}\xi_{1}-\lambda\sigma_{S}-\rho\sigma_{\mu})e^{2\xi_{1}(t-s)}},\\
g(t,s)= & \frac{1}{2}\log\Big(\frac{(\sigma_{S}\xi_{1}+\lambda\sigma_{S}+\rho\sigma_{\mu})+(\sigma_{S}\xi_{1}-\lambda\sigma_{S}-\rho\sigma_{\mu})e^{2\xi_{1}(t-s)}}{2\sigma_{S}\xi_{1}e^{\xi_{1}(t-s)}}\Big)\\
& -\frac{(1-p)(1-\rho^{2})}{2(1-p+p\rho^{2})}\log\Big(\frac{(\sigma_{S}\xi+\lambda\sigma_{S}+
\frac{\rho\sigma_{\mu}p}{1-p})+(\sigma_{S}\xi-\lambda\sigma_{S}-\frac{\rho\sigma_{\mu}p}{1-p})e^{2\xi(t-s)}}{2\sigma_{S}\xi e^{\xi(t-s)}} \Big)\\
& -\frac{\rho^{2}\lambda(s-t)}{2(1-p+p\rho^{2})}-\frac{\rho\sigma_{\mu}(s-t)}{2(1-p+p\rho^{2})\sigma_{S}},\nonumber
\end{split}
\end{equation}
\indent
The condition for the bounded Normal solution is
\begin{equation}\label{bod}
\gamma_{3}>0,\ \ \textrm{or}\ \ \gamma_{1}>0,\ \ \textrm{or}\ \ \gamma_{2}<0.
\end{equation}
\indent
The condition for the explosive solution and the critical point is
\begin{equation}
\gamma_{3}<0,\ \ \gamma_{1}<0,\ \ \textrm{and}\ \ \gamma_{2}>0,\ \ \ s-t=\frac{1}{2\xi}\log\Big(\frac{\gamma_{2}+\xi}{\gamma_{2}-\xi}\Big).\nonumber
\end{equation}
\begin{remark}\label{remapp}
If $p<0$, the conditions $(\ref{confffff})$ and $(\ref{bod})$ clearly hold, and we have $a(t,s)\leq{0}$ is a bounded solution as well as $1-2a(t,s)\hat{\Omega}(t)>1>0$ and $1-f(t,s)\hat{\Omega}(t)>1>0$. Hence we can conclude the solutions of ODEs $(\ref{odeA})$, $(\ref{odeB})$, $(\ref{odeC})$ are all bounded on $0\leq t\leq s\leq T$. We also notice that $A(t,s)=\frac{a(t,s)}{(1-p)(1-2a(t,s)\hat{\Omega}(t))}\leq{0}$, on $0\leq t\leq s\leq T$.
\end{remark}
\subsection{The Hyperbolic Solution}
The condition for the Hyperbolic solution is
\begin{equation}
\Delta\triangleq\lambda^{2}-\frac{2\lambda p\rho\sigma_{\mu}}{(1-p)\sigma^{2}_{S}}-\frac{p\sigma^{2}_{\mu}}{(1-p)\sigma^{2}_{S}}=0,\nonumber
\end{equation}
together with
\begin{equation}
\gamma_{2}\triangleq -\lambda+\frac{p\rho\sigma_{\mu}}{(1-p)\sigma_{S}}\neq 0,\nonumber
\end{equation}
\indent
Then we can solve $(\ref{eqna})$, $(\ref{eqnb})$, $(\ref{eqnc})$, $(\ref{eqnL})$ and $(\ref{eqnF})$ as:
\begin{equation}
\begin{split}
a(t,s)= & \frac{-1}{2\gamma_{1}(s-t-\frac{1}{\gamma_{2}})}-\frac{\gamma_{2}}{2\gamma_{1}},\\
b(t,s)= & -\frac{\lambda\bar{\mu}}{2\gamma_{1}\gamma_{2}(s-t-\frac{1}{\gamma_{2}})}-\frac{\gamma_{2}\lambda\bar{\mu}(s-t+\frac{1}{\gamma_{2}})}{2\gamma_{1}},\\
c(t,s)= & \frac{\gamma_{2}\sigma_{\mu}^{2}(s-t)}{2\gamma_{1}}+\frac{\lambda^{2}\bar{\mu}^{2}\gamma_{2}^{2}(s-t-\frac{4}{\gamma_{2}})(s-t)^{3}}{24\gamma_{1}(s-t-\frac{1}{\gamma_{2}})}+\frac{\sigma_{\mu}^{2}\log\Big|\frac{1}{2}(s-t)\gamma_{2}-1\Big|}{\gamma_{1}},\nonumber
\end{split}
\end{equation}
\begin{equation}
\begin{split}
f(t,s)= & -\frac{1}{2\sigma_{S}}\frac{1-e^{2\xi_{1}(t-s)}}{(\sigma_{S}\xi_{1}+\lambda\sigma_{S}+\rho\sigma_{\mu})+(\sigma_{S}\xi_{1}-\lambda\sigma_{S}-\rho\sigma_{\mu})e^{2\xi_{1}(t-s)}},\\
g(t,s)= & \frac{1}{2}\log\Big(\frac{(\sigma_{S}\xi_{1}+\lambda\sigma_{S}+\rho\sigma_{\mu})+(\sigma_{S}\xi_{1}-\lambda\sigma_{S}-\rho\sigma_{\mu})e^{2\xi_{1}(t-s)}}{2\sigma_{S}\xi_{1}e^{\xi_{1}(t-s)}}\Big)-\frac{(\lambda\sigma_{S}+\rho\sigma_{\mu})}{2\sigma_{S}}(s-t)\\
& +\frac{\sigma_{\mu}^{2}(1-\rho^{2})}{2\gamma_{1}}\Big[\log\Big|1+\gamma_{2}(t-s)\Big|-\gamma_{2}(s-t)\Big].\nonumber
\end{split}
\end{equation}
\indent
The condition for the bounded Hyperbolic solution is $\gamma_{2}<0$.

The condition for the explosive solution and the critical point is
\begin{equation}
\gamma_{2}>0, \ \ \ \textrm{and}\ \ \ s-t=\frac{1}{\gamma_{2}}.\nonumber
\end{equation}
\subsection{The Polynomial solution}
The condition for the Polynomial solution is
\begin{equation}
\Delta\triangleq\lambda^{2}-\frac{2\lambda p\rho\sigma_{\mu}}{(1-p)\sigma^{2}_{S}}-\frac{p\sigma^{2}_{\mu}}{(1-p)\sigma^{2}_{S}}=0,\nonumber
\end{equation}
together with
\begin{equation}
\gamma_{2}\triangleq -\lambda+\frac{p\rho\sigma_{\mu}}{(1-p)\sigma_{S}}= 0,\nonumber
\end{equation}
\indent
Then we can solve $(\ref{eqna})$, $(\ref{eqnb})$, $(\ref{eqnc})$, $(\ref{eqnL})$ and $(\ref{eqnF})$ as:
\begin{equation}
\begin{split}
a(t,s)= & \frac{p}{2(1-p)\sigma_{S}^{2}}(s-t),\\
b(t,s)= & \frac{p}{2(1-p)\sigma_{S}^{2}}\lambda\bar{\mu}(s-t)^{2},\\
c(t,s)= & -\frac{p}{4(1-p)\sigma_{S}^{2}}\sigma_{\mu}^{2}(s-t)^{2}+\frac{p}{6(1-p)\sigma_{S}^{2}}\lambda^{2}\bar{\mu}^{2}(s-t)^{3},\nonumber
\end{split}
\end{equation}
\begin{equation}
\begin{split}
f(t,s)= & -\frac{1}{2\sigma_{S}}\frac{1-e^{2\xi_{1}(t-s)}}{(\sigma_{S}\xi_{1}+\lambda\sigma_{S}+\rho\sigma_{\mu})+(\sigma_{S}\xi_{1}-\lambda\sigma_{S}-\rho\sigma_{\mu})e^{2\xi_{1}(t-s)}},\\
g(t,s)= & \frac{1}{2}\log\Big(\frac{(\sigma_{S}\xi_{1}+\lambda\sigma_{S}+\rho\sigma_{\mu})+(\sigma_{S}\xi_{1}-\lambda\sigma_{S}-\rho\sigma_{\mu})e^{2\xi_{1}(t-s)}}{2\sigma_{S}\xi_{1}e^{\xi_{1}(t-s)}}\Big) -\frac{(\lambda\sigma_{S}+\rho\sigma_{\mu})}{2\sigma_{S}}(s-t)\\
& -\frac{\sigma_{\mu}^{2}(1-\rho^{2})p}{4(1-p)\sigma_{S}^{2}}(s-t)^{2}.\nonumber
\end{split}
\end{equation}
\indent
All Polynomial solutions are bounded.
\subsection{The Tangent solution}
The condition for the Tangent solution is
\begin{equation}
\Delta\triangleq\lambda^{2}-\frac{2\lambda p\rho\sigma_{\mu}}{(1-p)\sigma^{2}_{S}}-\frac{p\sigma^{2}_{\mu}}{(1-p)\sigma^{2}_{S}}<0,\nonumber
\end{equation}
Now, we define
\begin{equation}
\zeta\triangleq \sqrt{-\Delta},\ \ \ \varpi\triangleq \tan^{-1}\Big(\frac{\gamma_{2}}{\zeta}\Big),\nonumber
\end{equation}
\indent
Then we can solve $(\ref{eqna})$, $(\ref{eqnb})$, $(\ref{eqnc})$, $(\ref{eqnL})$ and $(\ref{eqnF})$ as:
\begin{equation}
\begin{split}
a(t,s)= & \frac{\zeta}{2\gamma_{1}}\tan\Big(\zeta(s-t)+\varpi\Big)-\frac{\gamma_{2}}{2\gamma_{1}},\\
b(t,s)= & \frac{\lambda\bar{\mu}}{\gamma_{1}}\Big[-1-\tan(\varpi)\tan(\zeta(s-t)+\varpi)+\sec(\varpi)\sec(\zeta(s-t)+\varpi)\Big],\\
c(t,s)= & \frac{2\lambda^{2}\bar{\mu}^{2}\gamma_{2}\sqrt{\gamma_{2}^{2}+\zeta^{2}}}{2\gamma_{1}\zeta}\Big[\sec{(\varpi)}-\sec{(\zeta(s-t)+\varpi)}\Big]\\
& + \frac{\lambda^{2}\bar{\mu}^{2}(2\gamma_{2}+\zeta^{2})}{2\gamma_{1}\zeta^{3}}\Big[\tan(\zeta(s-t)+\varpi)-\tan(\varpi)\Big]\\
& -\frac{\lambda^{2}\bar{\mu}^{2}(\gamma_{2}^{2}+\zeta^{2})-\gamma_{2}\zeta^{2}\sigma_{\mu}^{2}}{2\gamma_{1}\zeta^{2}}+\frac{\sigma_{\mu}^{2}}{2\gamma_{1}}\log\Big(\sec(\varpi)\cos(\zeta(s-t)+\varpi)\Big),\nonumber
\end{split}
\end{equation}
\begin{equation}
\begin{split}
f(t,s)= & -\frac{1}{2\sigma_{S}}\frac{1-e^{2\xi_{1}(t-s)}}{(\sigma_{S}\xi_{1}+\lambda\sigma_{S}+\rho\sigma_{\mu})+(\sigma_{S}\xi_{1}-\lambda\sigma_{S}-\rho\sigma_{\mu})e^{2\xi_{1}(t-s)}},\\
g(t,s)= & \frac{1}{2}\log\Big(\frac{(\sigma_{S}\xi_{1}+\lambda\sigma_{S}+\rho\sigma_{\mu})+(\sigma_{S}\xi_{1}-\lambda\sigma_{S}-\rho\sigma_{\mu})e^{2\xi_{1}(t-s)}}{2\sigma_{S}\xi_{1}e^{\xi_{1}(t-s)}}\Big)  -\frac{(\lambda\sigma_{S}+\rho\sigma_{\mu})}{2\sigma_{S}}(s-t)\\
&-\sigma_{\mu}^{2}(1-\rho^{2})\Big[\frac{1}{2\gamma_{1}}\log\Big|\frac{\cos(\zeta(t-s)+\varpi)}{\cos(\varpi)}\Big|-\frac{\gamma_{2}}{2\gamma_{1}}(s-t)\Big].\nonumber
\end{split}
\end{equation}
\indent
All Tangent solutions are explosive solutions and the critical point is
\begin{equation}
s-t=\frac{\pi}{2\zeta}-\frac{1}{\zeta}\tan^{-1}(\frac{\gamma_{2}}{\zeta}).\nonumber
\end{equation}

\bibliography{diss}
\bibliographystyle{amsplain}

\end{document}